\documentclass[aps,physrev,reprint,superscriptaddress]{revtex4-2}

\usepackage{amsmath,amssymb,amsfonts,mathtools,amsthm}
\usepackage{bm}
\usepackage{braket}
\usepackage{physics}
\usepackage{mathrsfs}
\usepackage{dsfont}

\usepackage{graphicx}
\usepackage{xcolor}

\usepackage{comment}
\usepackage[linesnumbered,ruled,vlined]{algorithm2e}
\SetKw{Break}{break}

\newcommand{\sq}{\mathrm{SQ}}
\newcommand{\qsp}{\mathrm{Q}_{\mathrm{sp}}}

\newcommand{\rk}[1]{\operatorname{rank}\left(#1\right)}

\newcommand{\mc}[1]{\mathcal{#1}}
\newcommand{\mb}[1]{\mathbb{#1}}

\newcommand{\supp}{\operatorname{Supp}}

\newcommand{\diag}{\mathsf{diag}}
\newcommand{\card}{\operatorname{card}}

\newcommand{\flo}[1]{\left\|#1\right\|_{\mathrm{F}}}
\newcommand{\ope}[1]{\left\|#1\right\|_{\mathrm{op}}}

\newcommand{\dtv}{d_\mathrm{TV}}

\newcommand{\pr}{\operatorname{Pr}}

\newcommand{\tO}{\widetilde{O}}

\newtheorem{theorem}{Theorem}

\newtheorem{lemma}[theorem]{Lemma}
\newtheorem{definition}[theorem]{Definition}

\newtheorem{remark}[theorem]{Remark}

\newcommand{\poly}{\operatorname{poly}}
\newcommand{\polylog}{\operatorname{polylog}}

\usepackage[colorlinks=true,linkcolor=blue,urlcolor=blue,citecolor=blue]{hyperref}

\begin{document}


\title{A Quantum-Inspired Dequantization Method for Diagonally Weighted Matrix Functions: Application to Learning with Optimized Random Features
}



\author{Natsuto Isogai}
\email{natsutoisogai@g.ecc.u-tokyo.ac.jp}
\affiliation{Department of Physics, Graduate School of Science, The University of Tokyo, Tokyo, Japan}
\author{Mio Murao}
\email{murao@phys.s.u-tokyo.ac.jp}
\affiliation{Department of Physics, Graduate School of Science, The University of Tokyo, Tokyo, Japan}
\author{Hayata Yamasaki}
\email{hayata.yamasaki@gmail.com}
\affiliation{Department of Computer Science, Graduate School of Information Science and Technology, The University of Tokyo, Tokyo, Japan}

\begin{abstract}
Quantum-inspired classical algorithms have dequantized several quantum machine learning routines by replacing quantum linear-algebra subroutines with classical counterparts.
However, the sampler based on quantum singular value transformation (QSVT) for learning with optimized random features is not covered by existing dequantization frameworks, because the  matrix to be inverted is not itself available through sampling access.
In this work, we develop a classical algorithm to address this type of quantum-advantage candidate.
Our method samples heavy indices, reduces the transformation to a small principal block, and outputs a sparse classical representation with operator-norm guarantees.
Applying this method dequantizes the sampler for optimized random features, giving a classical sampler with prescribed accuracy and polynomially related runtime.
These results show that the factorization underlying a quantum block encoding can itself provide sufficient classical structure even when sampling-and-query access to the composite matrix is unavailable.
\end{abstract}


\maketitle


\textit{Introduction.}---
Machine learning has been a central tool for extracting structure from high-dimensional data.
Quantum machine learning (QML) aims to further accelerate such learning tasks by using quantum computation for linear-algebraic subroutines~\cite{wittek2014,PhysRevLett.103.150502,10.1098/rspa.2017.0551,PhysRevA.99.052331}.
Quantum singular value transformation (QSVT) provides a general framework for efficiently implementing spectral transformations of matrices acting on high-dimensional vector spaces~\cite{gilyen2019quantum,Martyn2021grand}.
In recent years, several works have proposed settings in which quantum learners can achieve provable advantages over classical learners, under complexity-theoretic or oracle assumptions~\cite{servedio2004,Liu2021,gyurik2023establishinglearningseparationsclassical,gyurik2024exponentialseparationsclassicalquantum,Yamasaki2026,Molteni2026}.
These results motivate a careful examination of which proposed QML advantages survive under comparable data-access models.

Quantum-inspired algorithms address this question by translating quantum state-preparation and block-encoding assumptions into classical sampling-and-query oracles~\cite{10.1145/3313276.3316310,chia2018quantuminspiredsublinearclassicalalgorithms,chia_et_al:LIPIcs:2020:13391,jethwani_et_al:LIPIcs.MFCS.2020.53,Gilyen2022improvedquantum,chia2022sampling,doi:10.1137/1.9781611977912.86,le2025robust,isogai2026winninglotteryticketsneural}.
Beginning with Tang's dequantization~\cite{10.1145/3313276.3316310} of the quantum recommendation system~\cite{kerenidis_et_al:LIPIcs.ITCS.2017.49}, several quantum linear-algebra and QML algorithms have been shown to admit efficient classical simulations under suitable input-model assumptions.
A common ingredient is sampling-and-query access to the matrix on which the spectral transformation is performed.
Such access permits entry queries, norm queries, and $\ell_2$-sampling from rows and columns, and thus supports randomized low-rank matrix approximations.
When this access is available, broad classes of QSVT-based algorithms can be simulated classically with polynomially related complexity.

However, sampling-and-query access to the transformed matrix is not implied by the oracle assumptions of every quantum algorithm.
A quantum algorithm may construct the required block-encoding by composing separately accessible operators, whereas the corresponding classical oracles need not provide sampling-and-query access to their product~\cite{10.5555/3495724.3496871,10.5555/3618408.3620034}.

The quantum algorithm for learning with optimized random features~\cite{10.5555/3495724.3496871} provides a representative example.
Optimized random features use a data-dependent sampling distribution designed to reduce the number of features required for kernel methods~\cite{JMLR:v18:15-178,pmlr-v70-avron17a,shahrampour2019samplingrandomfeaturesempirical,pmlr-v97-li19k,liu2020random}.
The quantum sampler uses QSVT to implement the inverse
\begin{equation}
    \bm{\Sigma}_{\epsilon}^{-1} = \left( \epsilon\bm I + \sqrt{\hat{\bm q}^{(\rho)}}\,\bm k\, \sqrt{\hat{\bm q}^{(\rho)}} \right)^{-1}, \label{eq:intro-inverse}
\end{equation}
where $\bm{I}$ denotes the identity matrix, $\hat{\bm q}^{(\rho)}$ is an operator representing an empirical probability distribution by its diagonal entries, and $\bm k$ is a kernel operator.
The input model supplies state preparation for $\sqrt{\hat{\bm q}^{(\rho)}}$ and query access to the spectral weights defining $\bm k$, from which $\bm k$ can be block encoded.
It does not supply sampling-and-query access to
$\sqrt{\hat{\bm q}^{(\rho)}}\bm k\sqrt{\hat{\bm q}^{(\rho)}}$.
In the discretized formulation of~\cite{10.5555/3495724.3496871}, this inverse acts on a space with $G^D$ basis states.
Under the stated oracle assumptions, the quantum sampler produces one optimized feature in time polynomial in the input dimension $D$ and the relevant precision parameters.
By contrast, a direct classical procedure that explicitly constructs and inverts the $G^D \times G^D$ matrix has cost polynomial in $G^D$ and is therefore exponential in $D$.

Existing dequantization techniques relevant to this inverse rely on one of two principal mechanisms.
Sampling-based methods construct randomized sketches using sampling access to the matrices to be transformed, as in dequantizations of recommendation systems~\cite{10.1145/3313276.3316310,chia2018quantuminspiredsublinearclassicalalgorithms,chia_et_al:LIPIcs:2020:13391,jethwani_et_al:LIPIcs.MFCS.2020.53,Gilyen2022improvedquantum,chia2022sampling,doi:10.1137/1.9781611977912.86}.
Sparse-matrix methods instead assume sparse-query access and are classically efficient for low-degree polynomial transformations at constant additive precision, as in the guided local-Hamiltonian problem~\cite{gharibian2022dequantizing}.

The optimized random-feature sampler is not a direct instance of either regime.
The sampling-based results cannot be applied directly because the corresponding classical input supplies sampling-and-query access only to the diagonal factor $\sqrt{\hat{\bm q}^{(\rho)}}$, and not to $\bm{k}$ or to the composite matrix $\sqrt{\hat{\bm q}^{(\rho)}}\,\bm k\, \sqrt{\hat{\bm q}^{(\rho)}}$.
The sparse-access result also cannot be invoked, since sparse-query access to the matrix to be transformed is not provided, and the target inverse is not restricted to the low-degree, constant-precision regime.

Nevertheless, here we prove that the absence of both structures does not prohibit the dequantization.
To establish this, we develop a general dequantization method for matrix functions of the form
\begin{equation}
    f(\gamma\bm I+\bm{D}\bm{A}\bm{D}),
    \label{eq: general-function}
\end{equation}
where the nonnegative diagonal matrix $\bm{D}$ is supplied through sampling-and-query access, whereas the Hermitian matrix $\bm{A}$, which may be dense and full rank, is supplied only through entry queries.
For optimized random features, taking $\bm{D}=\sqrt{\hat{\bm q}^{(\rho)}}$, $\bm{A}=\bm k$, and $f(x)=x^{-1}$ gives a classical optimized-random-feature sampler whose runtime is polynomial in $D$.

Unlike sampling-based dequantization methods, our algorithm does not require sampling-and-query access to $\bm{D} \bm{A} \bm{D}$, but uses only samples from the diagonal factor $\bm{D}$ to identify a small set of heavy coordinates.
In contrast to sparse-matrix methods, it imposes no sparsity condition on $\bm{A}$.
Once these coordinates are identified, entry queries to $\bm{A}$ on the corresponding indices suffice, even when $\bm{A}$ is dense and full rank.
The resulting algorithm approximates Eq.~\eqref{eq: general-function} and remains applicable when entry queries to $\bm{A}$ are available only through probabilistic estimates.

\textit{Problem setting.}---
Let $\bm{D}=\operatorname{diag}(d_1,\ldots,d_n)\geq0$ and let $\bm{A}\in\mathbb C^{n\times n}$ be Hermitian.
Sampling-and-query access $\mathrm{SQ}(\bm{D})$ allows one to query $d_i$, obtain the Frobenius norm $\|\bm{D}\|_{\mathrm F}$, and sample an index $i$ with probability $d_i^2/\|\bm{D}\|_{\mathrm F}^2$.
Entry-query access $\mathrm Q(\bm{A})$ returns $\bm{A}(i,j)$ for a requested pair $(i,j)$, but does not provide row or column sampling~\cite{10.1145/3313276.3316310,chia2018quantuminspiredsublinearclassicalalgorithms,chia_et_al:LIPIcs:2020:13391,jethwani_et_al:LIPIcs.MFCS.2020.53,Gilyen2022improvedquantum,chia2022sampling,doi:10.1137/1.9781611977912.86,le2025robust,gharibian2022dequantizing}.
These oracles make every entry of the weighted matrix directly accessible through
\begin{equation}
    (\bm{D}\bm{A}\bm{D})(i,j)=d_i\bm{A}(i,j)d_j,
    \label{eq:weighted-entry}
\end{equation}
but they do not generally implement $\mathrm{SQ}(\bm{D}\bm{A}\bm{D})$.
We also use a noisy query oracle $\mathrm{Q}_{\epsilon_A, \delta_A}(\bm{A})$, which returns every requested entry within additive error $\epsilon_A$ with probability at least $1 - \delta_A$.
Let $\mathbf q(\bm{D})$, $\mathbf s(\bm{D})$, and $\mathbf n(\bm{D})$ denote the runtimes of one entry query, one sample, and one norm query to $\mathrm{SQ}(\bm{D})$, respectively, and define
\begin{equation}
    \mathbf{sq}(\bm{D}) \coloneq \max\{\mathbf{q}(\bm{D}), \mathbf{s}(\bm{D}), \mathbf{n}(\bm{D})\}. \label{eq:sq-cost}
\end{equation}
Similarly, $\mathbf q(\bm{A})$ and $\mathbf q_{\epsilon_A,\delta_A}(\bm{A})$ denote the runtimes of one exact and one noisy entry query to $\bm{A}$, respectively.
Thus, $\mathrm{SQ}$ and $\mathrm Q$ denote access oracles, whereas $\mathbf{sq}$ and $\mathbf q$ denote the corresponding access costs.

For a real number $\gamma > 0$, define
\begin{equation}
    \bm{M} \coloneq \gamma\bm I+\bm{D}\bm{A}\bm{D}.
    \label{eq:matrix-M}
\end{equation}
Our objective is to construct a classical representation of a matrix $\bm N$ that approximates $f(\bm M)$ in the operator norm $\|\cdot\|_{\mathrm{op}}$ induced by the Euclidean norm, without exact sampling access to $\bm{D}\bm{A}\bm{D}$.

\begin{figure}[t]
    \centering
    \includegraphics[width=\linewidth]{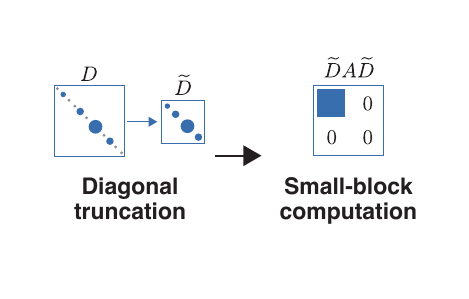}
    \caption{
    Diagonal truncation and small-block computation.
    Samples from $\mathrm{SQ}(\bm{D})$ identify the heavy coordinates and define the truncated matrix $\widetilde{\bm{D}}$.
    The matrix $\widetilde{\bm{D}}\bm{A}\widetilde{\bm{D}}$ is supported on the corresponding principal block, where the matrix function is computed exactly and efficiently.
    }
    \label{fig:diagonal-truncation}
\end{figure}

\textit{Diagonal-truncation method.}---
Our key scheme uses samples from $\mathrm{SQ}(\bm{D})$ to recover, with high probability, all coordinates whose diagonal entries exceed a chosen threshold.
Restricting $\bm{D}$ to these retained coordinates gives a truncated diagonal matrix $\widetilde{\bm{D}}$.
The matrix $\gamma \bm{I} + \widetilde{\bm{D}} \bm{A} \widetilde{\bm{D}}$ differs from $\gamma \bm{I}$ only on the principal block indexed by the retained coordinates.
Therefore, we query $\bm{A}$ only on this block and compute the matrix function there, while the remaining coordinates contribute $f(\gamma) \bm{I}$.
The construction is also illustrated in Fig.~\ref{fig:diagonal-truncation}.

We now formalize this construction and then state the resulting approximation guarantee.
Let $\epsilon_D>0$ be a truncation threshold and let $\delta_D\in(0,1)$ be the allowed failure probability of the truncation step.
Define the set of heavy coordinates by
\begin{equation}
    H_{\epsilon_D}=\{i\in[n]:d_i\geq\epsilon_D\}.
    \label{eq:heavy-set}
\end{equation}
Every $i\in H_{\epsilon_D}$ is sampled by $\mathrm{SQ}(\bm{D})$ with probability at least $\epsilon_D^2/\|\bm{D}\|_{\mathrm F}^2$.
Moreover,
\begin{equation}
    |H_{\epsilon_D}|
    \leq
    \frac{\|\bm{D}\|_{\mathrm F}^2}{\epsilon_D^2},
    \label{eq:heavy-set-size}
\end{equation}
because the squared diagonal entries sum to $\|\bm{D}\|_{\mathrm F}^2$.
It follows that
\begin{equation}
    R=
    \left\lceil
    \frac{\|\bm{D}\|_{\mathrm F}^2}{\epsilon_D^2}
    \log\left(
    \frac{\|\bm{D}\|_{\mathrm F}^2}
    {\epsilon_D^2\delta_D}
    \right)
    \right\rceil
    \label{eq:number-samples}
\end{equation}
independent samples recover all coordinates in $H_{\epsilon_D}$ with probability at least $1-\delta_D$.

Let $S$ be the recovered set, write $s=|S|$, and let $\widetilde{\bm{D}}$ retain $d_i$ on $S$ and vanish elsewhere.
On the successful event,
\begin{equation}
    \|\bm{D}-\widetilde{\bm{D}}\|_{\mathrm{op}}
    \leq\epsilon_D,
    \qquad
    s=\operatorname{rank}(\widetilde{\bm{D}})
    \leq
    \frac{\|\bm{D}\|_{\mathrm F}^2}{\epsilon_D^2}.
    \label{eq:truncation}
\end{equation}
The reduced dimension is therefore determined by the heavy-coordinate support of $\bm{D}$.
No low-rank or sparsity condition is imposed on $\bm{A}$.

Choose an ordering $S=\{i_1,\ldots,i_s\}$ and define the bijection $\pi(i_t)=t$.
The algorithm queries the upper-triangular entries of $\bm{A}$ on $S\times S$ using $\mathrm Q_{\epsilon_A,\delta_A}(\bm{A})$ and enforces Hermitian symmetry on the returned values.
It then constructs the $s\times s$ matrix
\begin{equation}
    \widetilde{\bm M}_S(t,u) =\gamma\delta_{t,u} +d_{i_t}\widetilde{\bm{A}}(i_t,i_u)d_{i_u}\label{eq:retained-block}
\end{equation}
and computes
\begin{equation}
    \bm B=f(\widetilde{\bm M}_S)\label{eq:block-function}
\end{equation}
by diagonalizing $\widetilde{\bm M}_S$ and applying $f$ to its eigenvalues.
The output matrix is represented by $S$, $\pi$, and $\bm B$ as
\begin{equation}
    \bm N(i,j)=
    \begin{cases}
        \bm B(\pi(i),\pi(j)),&i,j\in S,\\
        f(\gamma),&i=j\notin S,\\
        0,&\text{otherwise}.
    \end{cases}
    \label{eq:N-structure}
\end{equation}
Indeed, after permuting the coordinates in $S$ to the first $s$ positions,
\begin{equation}
    \gamma\bm I+ \widetilde{\bm{D}}\widetilde{\bm{A}}\widetilde{\bm{D}} \simeq
    \begin{pmatrix}
        \widetilde{\bm M}_S&0\\
        0&\gamma\bm I_{n-s}
    \end{pmatrix}.
    \label{eq:block}
\end{equation}
Its matrix function has the block $\bm B$ on $S$ and equals $f(\gamma)\bm I$ on the complement.
Thus, $\bm N$ is stored using the index set $S$ and one $s\times s$ block rather than an $n\times n$ array.

\textit{Main result.}---
Let $\widetilde{\bm{A}}$ denote any Hermitian extension of the queried entries on $S\times S$, and define
\begin{equation}
    \widetilde{\bm M} = \gamma\bm I+ \widetilde{\bm{D}}\widetilde{\bm{A}}\widetilde{\bm{D}}. \label{eq:truncated-M}
\end{equation}
On the event that the truncation step and all retained-block queries succeed,
\begin{align}
    \|\bm M-\widetilde{\bm M}\|_{\mathrm{op}} &\leq 2\|\bm{A}\|_{\mathrm{op}} \|\bm{D}\|_{\mathrm{op}}\epsilon_D + \epsilon_A\|\bm{D}\|_{\mathrm F}^{2} \notag\\
    &\eqqcolon \nu.\label{eq:nu}
\end{align}
The first term comes from truncating the two diagonal factors, and the second comes from estimating the entries of $\bm{A}$ on the retained block.

The preceding approximation controls the error in the matrix to which $f$ is applied.
To obtain an approximation guarantee after applying the matrix function $f$, we need to bound $\| f(\bm X) - f(\bm Y) \|_\mathrm{op}$ in terms of $\| \bm{X} - \bm{Y} \|_{\mathrm{op}}$.
A scalar function $f$ is $L$-Lipschitz on an interval $\mathfrak{F} \subseteq \mathbb{R}$ if
\begin{equation}
    |f(x)-f(y)| \leq L|x-y| \qquad (x,y\in \mathfrak{F}).
\end{equation}
For Hermitian matrices with spectra contained in $\mathfrak{F}$, the corresponding operator-norm perturbation bound incurs an additional logarithmic factor in the number of distinct eigenvalues~\cite{ALEKSANDROV20112741}.
Since the approximating matrix in Eq.~\eqref{eq:block} has at most $s+1$ distinct eigenvalues, this gives the factor $C[1+\log(1+s)]L$ appearing below.

The following theorem summarizes the size of the retained block, the resulting approximation error, the success probability, and the running time.
The detailed proofs are given in Appendix~\ref{sec: main theorems}.
\begin{theorem}[Diagonal-truncation approximation]
\label{thm:main}
Suppose that $f:\mathbb R\to\mathbb R$ is $L$-Lipschitz on a spectral interval containing
\begin{equation}
    [\lambda_{\min}(\bm M)-\nu, \lambda_{\max}(\bm M)+\nu]\cup\{\gamma\}. \label{eq:lipschitz-interval}
\end{equation}
Given $\mathrm{SQ}(\bm{D})$ and
$\mathrm Q_{\epsilon_A,\delta_A}(\bm{A})$, the randomized classical algorithm described above outputs the representation of $\bm N$ in Eq.~\eqref{eq:N-structure}.
With probability at least
$1-\delta_D-s(s+1)\delta_A/2$, the output satisfies
\begin{align}
    s&\leq \frac{\|\bm{D}\|_{\mathrm F}^{2}}{\epsilon_D^{2}}, \label{eq:s-bound}\\
    \|f(\bm M)-\bm N\|_{\mathrm{op}} &\leq C[1+\log(1+s)]L\nu, \label{eq:main-error}
\end{align}
where $C>0$ is a numerical constant.
The runtime is
\begin{align}
    &O\left[\mathbf{sq}(\bm{D}) +\mathbf q_{\epsilon_A,\delta_A}(\bm{A})+1 \right] \notag\\
    &\qquad\times \widetilde O\left[ \left(\frac{\|\bm{D}\|_{\mathrm F}}{\epsilon_D}\right)^6\log\left(\frac1{\delta_D}\right)\right].\label{eq:main-runtime}
\end{align}
\end{theorem}

\textit{Optimized random features.}---
Random-feature methods approximate a shift-invariant kernel by sampling a finite collection of Fourier features and fitting a linear predictor in the resulting feature space~\cite{NIPS2007_013a006f, NIPS2008_0efe3284}.
Once the feature parameters have been sampled, the remaining learning problem is an ordinary classical linear regression problem.
In the standard construction, the feature parameters are sampled from the spectral distribution defining the kernel.
Optimized random features instead use a data-dependent distribution related to regularized leverage scores~\cite{JMLR:v18:15-178,pmlr-v70-avron17a,shahrampour2019samplingrandomfeaturesempirical,pmlr-v97-li19k,liu2020random}.
This distribution assigns larger probability to features that are more useful for the observed input distribution and can reduce the number of features required to achieve the approximation quality of kernel ridge regression.

The quantum algorithm in~\cite{10.5555/3495724.3496871} addresses the computational problem of sampling from this optimized distribution.
The coefficient-fitting stage remains classical.
We now apply Theorem~\ref{thm:main} to the quantum sampling subroutine.

Let $\mathcal X=\{0,\ldots,G-1\}^D$, so that $|\mathcal X|=G^D$, and let $\bm F_D$ denote the $D$-dimensional discrete Fourier transform over $\mathcal X$.
The shift-invariant kernel is specified by a nonnegative diagonal matrix $\bm Q^{(\tau)}$ through
\begin{equation}
    \bm k=\bm F_D^\dagger\bm Q^{(\tau)}\bm F_D, \qquad Q_{\max}^{(\tau)}=\|\bm Q^{(\tau)}\|_{\mathrm{op}}.\label{eq:orf-kernel}
\end{equation}
The empirical distribution $\hat q^{(\rho)}$ is represented both as a probability function over $\mathcal X$ and as the diagonal matrix $\hat{\bm q}^{(\rho)}$, and we define
\begin{equation}
    \bm\Sigma_\epsilon =\epsilon\bm I +\sqrt{\hat{\bm q}^{(\rho)}}\,\bm k\,\sqrt{\hat{\bm q}^{(\rho)}}.\label{eq:orf-sigma}
\end{equation}
For each $x\in\mathcal X$, let
\begin{equation}
    \ket{\phi_x} =\sqrt{\hat{\bm q}^{(\rho)}}\,\bm F_D\sqrt{\bm Q^{(\tau)}}\ket{x}. \label{eq:orf-feature-vector}
\end{equation}
The optimized distribution over random features is then given by
\begin{equation}
    w(x)=\bra{\phi_x}\bm\Sigma_\epsilon^{-1}\ket{\phi_x}, \qquad \mathcal D_\epsilon(x)=\frac{w(x)}{\sum_{z\in\mathcal X}w(z)}.\label{eq:orf-distribution}
\end{equation}
The quantum algorithm prepares a state whose measurement approximately samples from $\mathcal D_\epsilon$.
Under its oracle assumptions, one sample with total-variation error at most $\delta$ is produced in time
\begin{align}
    &O\left[D\log G\log\log G+\mathbf q(\bm Q^{(\tau)})+\mathbf{sq}\left(\sqrt{\hat{\bm q}^{(\rho)}}\right)\right]\notag\\
    &\quad\times\widetilde O\left[\frac{Q_{\max}^{(\tau)}}{\epsilon}\polylog\left(\frac{1}{\delta}\right)\right].\label{eq:orf-quantum-runtime}
\end{align}
The central computational step of the quantum sampler is the regularized inverse $\bm \Sigma_\epsilon^{-1}$.
To compare the quantum and classical algorithms under matched access assumptions, we give the classical algorithm query access $\mathrm Q(\bm Q^{(\tau)})$, sampling-and-query access $\mathrm{SQ}(\sqrt{\hat{\bm q}^{(\rho})})$, and the value $Q_{\max}^{(\tau)}$.
These are the classical counterparts of the kernel-description oracle and the empirical-distribution state-preparation oracle used by the quantum algorithm~\cite{10.5555/3495724.3496871}.
Their access costs are denoted by $\mathbf q(\bm{Q}^{(\tau)})$ and $\mathbf{sq}(\sqrt{\hat{\bm{q}}^{(\rho)}})$, respectively.

Under this access model, the classical input does not provide the sampling access or sparse access to $\sqrt{\hat{\bm q}^{(\rho)}}\bm k \sqrt{\hat{\bm q}^{(\rho)}}$ assumed by existing dequantization methods~\cite{10.1145/3313276.3316310,chia2018quantuminspiredsublinearclassicalalgorithms,chia_et_al:LIPIcs:2020:13391,jethwani_et_al:LIPIcs.MFCS.2020.53,Gilyen2022improvedquantum,chia2022sampling,doi:10.1137/1.9781611977912.86,gharibian2022dequantizing}.
Nevertheless, the regularized inverse in Eq.~\eqref{eq:orf-sigma} lies precisely in the regime addressed by Theorem~\ref{thm:main}, and the available classical oracles satisfy its access requirements.

Sampling-and-query access to $\sqrt{\hat{\bm q}^{(\rho)}}$ directly supplies the nonnegative diagonal factor.
Although entry-query access to the kernel matrix $\bm{k}$ is not given directly, its Fourier representation in Eq.~\eqref{eq:orf-kernel} provides an efficient algorithm for estimating each entry by using $\mathrm Q(\bm Q^{(\tau)})$.
Moreover, since $\hat q^{(\rho)}$ is a normalized probability distribution,
\begin{equation}
    \left\|\sqrt{\hat{\bm q}^{(\rho)}}\right\|_{\mathrm F}=1.
\end{equation}
Therefore, the Frobenius-norm dependence in Theorem~\ref{thm:main} introduces no factor that grows with $|\mathcal{X}| = G^D$.

Thus all conditions needed to apply Theorem~\ref{thm:main} can be implemented under the corresponding classical input, despite the absence of sampling-and-query access to $\sqrt{\hat{\bm q}^{(\rho)}}\bm k \sqrt{\hat{\bm q}^{(\rho)}}$.
The theorem replaces this unavailable matrix-level access by sampling from the diagonal factor and computing only on the corresponding principal block as illustrated in Fig.~\ref{fig:diagonal-truncation}.

The resulting principal-block representation can be used to approximate the optimized feature weights and to sample from the resulting distribution without enumerating the full domain.
The detailed proof is given in Appendix~\ref{sec: proof of Rejection sampling from approximated distribution}.
We therefore obtain the following result.

\begin{theorem}[Classical optimized-feature sampler]
\label{thm:orf-sampler}
Given query access $\mathrm Q(\bm Q^{(\tau)})$, sampling-and-query access $\mathrm{SQ}(\sqrt{\hat{\bm q}^{(\rho)}})$, and the value $Q_{\max}^{(\tau)}$, and assuming $\bm k(0,0)=\Omega(1)$ as in Ref.~\cite{10.5555/3495724.3496871}, for $0<\epsilon\leq \min\left\{Q_{\max}^{(\tau)},\sqrt{Q_{\max}^{(\tau)}}\right\}$, and $\delta\in(0,1)$, there exists a randomized classical algorithm that outputs a sample from a distribution $\mathcal D'$ satisfying
\begin{equation}
    d_{\mathrm{TV}}(\mathcal D_\epsilon,\mathcal D')\leq\delta.
    \label{eq:orf-tv}
\end{equation}
After the classical data structure for the empirical distribution has been constructed, the runtime for producing one sample is
\begin{align}
    &O\left[D\polylog G+\mathbf q(\bm Q^{(\tau)})+\mathbf{sq}\left(\sqrt{\hat{\bm q}^{(\rho)}}\right)\right] \notag\\
    &\quad\times\widetilde O\left[\left(\frac{Q_{\max}^{(\tau)}}{\epsilon}\right)^{18}\frac{1}{\delta^6}\right].\label{eq:orf-classical-runtime}
\end{align}
\end{theorem}

The parameter dependences in Eqs.~\eqref{eq:orf-quantum-runtime} and~\eqref{eq:orf-classical-runtime} do not match, and the powers in our classical bound are not claimed to be optimal.
In particular, $\delta$ describes the total-variation error of a single feature sample and does not by itself determine the end-to-end statistical cost, which also involves repeated sampling.
The relevant complexity conclusion is that, when $Q_{\max}^{(\tau)}/\epsilon=\poly(D)$ and $1/\delta=\poly(D)$, the classical sampler runs in time polynomial in $D$ and $\log G$ under the matched oracle assumptions.
Hence, the proposed superpolynomial separation in $D$ disappears.
Since optimized-feature sampling is the only quantum component of the learning procedure in Theorem~2 of~\cite{10.5555/3495724.3496871}, and the subsequent coefficient fitting is classical, replacing the sampler also yields a fully classical implementation of the learning procedure.
This is an oracle-level dequantization of the QSVT-based learning procedure rather than a new statistical method for random-feature learning.
It shows that the absence of sampling-and-query access to the composite matrix $\sqrt{\hat{\bm q}^{(\rho)}}\,\bm k\,
\sqrt{\hat{\bm q}^{(\rho)}}$ does not prevent a polynomial-time classical simulation under the corresponding oracle assumptions.

\textit{Conclusion.}---
In this work, we have introduced a dequantization method for diagonally weighted matrix functions that operates without sampling access to the composite matrix itself.
Given sampling-and-query access to a nonnegative diagonal factor $\bm{D}$ and entry-query access to a Hermitian factor $\bm{A}$, the method identifies an effective coordinate support from $\bm{D}$ and reduces the matrix transformation to a small principal block.
This yields a support-sparse approximation to $f(\gamma\bm I+\bm{D}\bm{A}\bm{D})$ with an operator-norm guarantee.
Applied to optimized random features, it yields a classical sampler running in polynomial time under the parameter regime and matched oracle assumptions.
Since the subsequent coefficient fitting is classical, this also gives a fully classical implementation of the learning procedure in~\cite{10.5555/3495724.3496871} and removes the proposed superpolynomial separation in the input dimension.

These results suggest that the dequantization of QSVT-based routines should be assessed through the combination of access structure and target transformation.
Our result identifies diagonal weighting together with entry-wise access to the remaining Hermitian factor as a sufficient structure for dequantization.
More broadly, the absence of sampling access to a block-encoded composite matrix does not by itself establish a quantum advantage, since the separately accessible factors may still enable a different efficient classical reduction.
Characterizing more general factor-level access conditions is a natural direction for future work.

\begin{acknowledgments}
We acknowledge the use of ChatGPT 5.6 sol to assist in the preparation of the manuscript.
NI was supported by JST BOOST, Japan Grant Number JPMJBS2418, JST CREST Grant Number JPMJCR25I5, JSPS KAKENHI Grant Number 23K21643, and MEXT Quantum Leap Flagship Program (MEXT QLEAP) JPMXS0118069605, JPMXS0120351339.
MM was supported by MEXT Quantum Leap Flagship Program (MEXT QLEAP) JPMXS0118069605, JPMXS0120351339, JST CREST Grant Number JPMJCR25I5, JST ASPIRE Grant Number JPMJAP25A3, JSPS KAKENHI Grant Number 23K21643, JST NEXUS Grant Number JPMJNX26C9, and IBM Quantum.
HY was supported by JST PRESTO Grant Number JPMJPR201A, JPMJPR23FC, JSPS KAKENHI Grant Number JP23K19970, JST CREST Grant Number JPMJCR25I5, JST [Moonshot R\&D] [Grant Number JPMJMS256J], and Faculty Research Funding from Google Quantum AI\@.
\end{acknowledgments}

\bibliography{main}

@article{chia2022sampling,
author = {Chia, Nai-Hui and Gily\'{e}n, Andr\'{a}s Pal and Li, Tongyang and Lin, Han-Hsuan and Tang, Ewin and Wang, Chunhao},
title = {Sampling-based Sublinear Low-rank Matrix Arithmetic Framework for Dequantizing Quantum Machine Learning},
year = {2022},
issue_date = {October 2022},
publisher = {Association for Computing Machinery},
address = {New York, NY, USA},
volume = {69},
number = {5},
issn = {0004-5411},
url = {https://doi.org/10.1145/3549524},
doi = {10.1145/3549524},
journal = {J. ACM},
month = oct,
articleno = {33},
numpages = {72}
}

@inproceedings{gilyen2019quantum,
author = {Gily\'{e}n, Andr\'{a}s and Su, Yuan and Low, Guang Hao and Wiebe, Nathan},
title = {Quantum singular value transformation and beyond: exponential improvements for quantum matrix arithmetics},
year = {2019},
isbn = {9781450367059},
publisher = {Association for Computing Machinery},
address = {New York, NY, USA},
url = {https://doi.org/10.1145/3313276.3316366},
doi = {10.1145/3313276.3316366},
booktitle = {Proceedings of the 51st Annual ACM SIGACT Symposium on Theory of Computing},
pages = {193–204},
numpages = {12},
location = {Phoenix, AZ, USA},
series = {STOC 2019}
}

@article{le2025robust,
  author       = {Le Gall, Fran{\c{c}}ois},
  title        = {Robust Dequantization of the Quantum Singular Value Transformation and Quantum Machine Learning Algorithms},
  journal      = {computational complexity},
  year         = {2025},
  volume       = {34},
  number       = {1},
  pages        = {2},
  doi          = {10.1007/s00037-024-00262-3},
  url          = {https://link.springer.com/article/10.1007/s00037-024-00262-3}
}

@inproceedings{gharibian2022dequantizing,
author = {Gharibian, Sevag and Le Gall, Fran\c{c}ois},
title = {Dequantizing the Quantum singular value transformation: hardness and applications to Quantum chemistry and the Quantum {PCP} conjecture},
year = {2022},
isbn = {9781450392648},
publisher = {Association for Computing Machinery},
address = {New York, NY, USA},
url = {https://doi.org/10.1145/3519935.3519991},
doi = {10.1145/3519935.3519991},
booktitle = {Proceedings of the 54th Annual ACM SIGACT Symposium on Theory of Computing},
pages = {19–32},
numpages = {14},
location = {Rome, Italy},
series = {STOC 2022}
}

@article{ALEKSANDROV20112741,
title = {Estimates of operator moduli of continuity},
journal = {Journal of Functional Analysis},
volume = {261},
number = {10},
pages = {2741-2796},
year = {2011},
issn = {0022-1236},
doi = {doi.org/10.1016/j.jfa.2011.07.009},
url = {https://www.sciencedirect.com/science/article/pii/S0022123611002631},
author = {A.B. Aleksandrov and V.V. Peller}
}

@article{PhysRevLett.103.150502,
  title = {Quantum Algorithm for Linear Systems of Equations},
  author = {Harrow, Aram W. and Hassidim, Avinatan and Lloyd, Seth},
  journal = {Phys. Rev. Lett.},
  volume = {103},
  issue = {15},
  pages = {150502},
  numpages = {4},
  year = {2009},
  month = {Oct},
  publisher = {American Physical Society},
  doi = {10.1103/PhysRevLett.103.150502},
  url = {https://link.aps.org/doi/10.1103/PhysRevLett.103.150502}
}

@inproceedings{10.5555/3495724.3496871,
  author    = {Yamasaki, Hayata and Subramanian, Sathyawageeswar and Sonoda, Sho and Koashi, Masato},
  title     = {Learning with Optimized Random Features: Exponential Speedup by Quantum Machine Learning without Sparsity and Low-Rank Assumptions},
  booktitle = {Advances in Neural Information Processing Systems 33},
  pages     = {13674--13687},
  isbn = {9781713829546},
  year      = {2020},
  url       = {https://proceedings.neurips.cc/paper/2020/hash/9ddb9dd5d8aee9a76bf217a2a3c54833-Abstract.html},
}

@book{scholkopf2002learning,
  title={Learning with kernels: support vector machines, regularization, optimization, and beyond},
  author={Sch{\"o}lkopf, Bernhard and Smola, Alexander J},
  year={2002},
  publisher={MIT press},
  url = {https://doi.org/10.7551/mitpress/4175.001.0001},
}

@inproceedings{NIPS2007_013a006f,
 author = {Rahimi, Ali and Recht, Benjamin},
 booktitle = {Advances in Neural Information Processing Systems},
 editor = {J. Platt and D. Koller and Y. Singer and S. Roweis},
 pages = {},
 publisher = {Curran Associates, Inc.},
 title = {Random Features for Large-Scale Kernel Machines},
 url = {https://proceedings.neurips.cc/paper_files/paper/2007/file/013a006f03dbc5392effeb8f18fda755-Paper.pdf},
 volume = {20},
 year = {2007}
}

@inproceedings{NIPS2008_0efe3284,
 author = {Rahimi, Ali and Recht, Benjamin},
 booktitle = {Advances in Neural Information Processing Systems},
 editor = {D. Koller and D. Schuurmans and Y. Bengio and L. Bottou},
 pages = {},
 publisher = {Curran Associates, Inc.},
 title = {Weighted Sums of Random Kitchen Sinks: Replacing minimization with randomization in learning},
 url = {https://proceedings.neurips.cc/paper_files/paper/2008/file/0efe32849d230d7f53049ddc4a4b0c60-Paper.pdf},
 volume = {21},
 year = {2008}
}

@article{brent1976fast,
author = {Brent, Richard P.},
title = {Fast Multiple-Precision Evaluation of Elementary Functions},
year = {1976},
issue_date = {April 1976},
publisher = {Association for Computing Machinery},
address = {New York, NY, USA},
volume = {23},
number = {2},
issn = {0004-5411},
url = {https://doi.org/10.1145/321941.321944},
doi = {10.1145/321941.321944},
journal = {J. ACM},
month = apr,
pages = {242–251},
numpages = {10}
}

@InProceedings{haible1998fast,
author="Haible, Bruno
and Papanikolaou, Thomas",
editor="Buhler, Joe P.",
title="Fast multiprecision evaluation of series of rational numbers",
booktitle="Algorithmic Number Theory",
year="1998",
publisher="Springer Berlin Heidelberg",
address="Berlin, Heidelberg",
pages="338--350",
isbn="978-3-540-69113-6",
doi       = {10.1007/BFb0054873},
url       = {https://link.springer.com/chapter/10.1007/BFb0054873}
}

@book{brent2010modern,
  title     = {Modern Computer Arithmetic},
  author    = {Brent, Richard P. and Zimmermann, Paul},
  year      = {2010},
  publisher = {Cambridge University Press},
  volume    = {18},
  series    = {Cambridge Monographs on Applied and Computational Mathematics},
  doi       = {10.1017/CBO9780511921698},
  url       = {https://www.cambridge.org/core/books/modern-computer-arithmetic/B345DF9E08B232BC5BAFE05C2F3A36D8}
}

@InProceedings{10.5555/3618408.3620034,
  title = 	 {Quantum Ridgelet Transform: Winning Lottery Ticket of Neural Networks with Quantum Computation},
  author =       {Yamasaki, Hayata and Subramanian, Sathyawageeswar and Hayakawa, Satoshi and Sonoda, Sho},
  booktitle = 	 {Proceedings of the 40th International Conference on Machine Learning},
  pages = 	 {39008--39034},
  year = 	 {2023},
  editor = 	 {Krause, Andreas and Brunskill, Emma and Cho, Kyunghyun and Engelhardt, Barbara and Sabato, Sivan and Scarlett, Jonathan},
  volume = 	 {202},
  series = 	 {Proceedings of Machine Learning Research},
  month = 	 {23--29 Jul},
  publisher =    {PMLR},
  url = 	 {https://proceedings.mlr.press/v202/yamasaki23a.html}
}

@InProceedings{kerenidis_et_al:LIPIcs.ITCS.2017.49,
  author =	{Kerenidis, Iordanis and Prakash, Anupam},
  title =	{{Quantum Recommendation Systems}},
  booktitle =	{8th Innovations in Theoretical Computer Science Conference (ITCS 2017)},
  pages =	{49:1--49:21},
  series =	{Leibniz International Proceedings in Informatics (LIPIcs)},
  ISBN =	{978-3-95977-029-3},
  ISSN =	{1868-8969},
  year =	{2017},
  volume =	{67},
  editor =	{Papadimitriou, Christos H.},
  publisher =	{Schloss Dagstuhl -- Leibniz-Zentrum f{\"u}r Informatik},
  address =	{Dagstuhl, Germany},
  URL =		{https://drops.dagstuhl.de/entities/document/10.4230/LIPIcs.ITCS.2017.49},
  URN =		{urn:nbn:de:0030-drops-81541},
  doi =		{10.4230/LIPIcs.ITCS.2017.49}
}

@article{JMLR:v18:15-178,
  author  = {Francis Bach},
  title   = {On the Equivalence between Kernel Quadrature Rules and Random Feature Expansions},
  journal = {Journal of Machine Learning Research},
  year    = {2017},
  volume  = {18},
  number  = {21},
  pages   = {1--38},
  url     = {http://jmlr.org/papers/v18/15-178.html}
}

@inproceedings{10.1145/3313276.3316310,
author = {Tang, Ewin},
title = {A quantum-inspired classical algorithm for recommendation systems},
year = {2019},
isbn = {9781450367059},
publisher = {Association for Computing Machinery},
address = {New York, NY, USA},
url = {https://doi.org/10.1145/3313276.3316310},
doi = {10.1145/3313276.3316310},
booktitle = {Proceedings of the 51st Annual ACM SIGACT Symposium on Theory of Computing},
pages = {217–228},
numpages = {12},
location = {Phoenix, AZ, USA},
series = {STOC 2019}
}

@Article{Liu2021,
author={Liu, Yunchao
and Arunachalam, Srinivasan
and Temme, Kristan},
title={A rigorous and robust quantum speed-up in supervised machine learning},
journal={Nature Physics},
year={2021},
month={Sep},
day={01},
volume={17},
number={9},
pages={1013-1017},
issn={1745-2481},
doi={10.1038/s41567-021-01287-z},
url={https://doi.org/10.1038/s41567-021-01287-z}
}

@article{servedio2004,
author = {Servedio, Rocco A. and Gortler, Steven J.},
title = {Equivalences and Separations Between Quantum and Classical Learnability},
journal = {SIAM Journal on Computing},
volume = {33},
number = {5},
pages = {1067-1092},
year = {2004},
doi = {10.1137/S0097539704412910},

URL = { 
    
        https://doi.org/10.1137/S0097539704412910
    
    

}
}

@misc{gyurik2023establishinglearningseparationsclassical,
      title={On establishing learning separations between classical and quantum machine learning with classical data}, 
      author={Casper Gyurik and Vedran Dunjko},
      year={2023},
      eprint={2208.06339},
      archivePrefix={arXiv},
      primaryClass={quant-ph},
      url={https://arxiv.org/abs/2208.06339}, 
}

@misc{gyurik2024exponentialseparationsclassicalquantum,
      title={Exponential separations between classical and quantum learners}, 
      author={Casper Gyurik and Vedran Dunjko},
      year={2024},
      eprint={2306.16028},
      archivePrefix={arXiv},
      primaryClass={quant-ph},
      url={https://arxiv.org/abs/2306.16028}, 
}

@Article{Yamasaki2026,
author={Yamasaki, Hayata
and Isogai, Natsuto
and Murao, Mio},
title={Advantage of quantum machine learning from general computational advantages},
journal={npj Quantum Information},
year={2026},
month={Jul},
day={28},
volume={12},
number={1},
pages={125},
issn={2056-6387},
doi={10.1038/s41534-026-01279-y},
url={https://doi.org/10.1038/s41534-026-01279-y}
}

@Article{Molteni2026,
author={Molteni, Riccardo
and Gyurik, Casper
and Dunjko, Vedran},
title={Exponential quantum advantages in learning quantum observables from classical data},
journal={npj Quantum Information},
year={2026},
month={Jan},
day={10},
volume={12},
number={1},
pages={19},
issn={2056-6387},
doi={10.1038/s41534-025-01162-2},
url={https://doi.org/10.1038/s41534-025-01162-2}
}

@InProceedings{jethwani_et_al:LIPIcs.MFCS.2020.53,
  author =	{Jethwani, Dhawal and Le Gall, Fran\c{c}ois and Singh, Sanjay K.},
  title =	{{Quantum-Inspired Classical Algorithms for Singular Value Transformation}},
  booktitle =	{45th International Symposium on Mathematical Foundations of Computer Science (MFCS 2020)},
  pages =	{53:1--53:14},
  series =	{Leibniz International Proceedings in Informatics (LIPIcs)},
  ISBN =	{978-3-95977-159-7},
  ISSN =	{1868-8969},
  year =	{2020},
  volume =	{170},
  editor =	{Esparza, Javier and Kr\'{a}l', Daniel},
  publisher =	{Schloss Dagstuhl -- Leibniz-Zentrum f{\"u}r Informatik},
  address =	{Dagstuhl, Germany},
  URL =		{https://drops.dagstuhl.de/entities/document/10.4230/LIPIcs.MFCS.2020.53},
  URN =		{urn:nbn:de:0030-drops-127193},
  doi =		{10.4230/LIPIcs.MFCS.2020.53}
}

@InProceedings{pmlr-v70-avron17a,
  title = 	 {Random {F}ourier Features for Kernel Ridge Regression: Approximation Bounds and Statistical Guarantees},
  author =       {Haim Avron and Michael Kapralov and Cameron Musco and Christopher Musco and Ameya Velingker and Amir Zandieh},
  booktitle = 	 {Proceedings of the 34th International Conference on Machine Learning},
  pages = 	 {253--262},
  year = 	 {2017},
  editor = 	 {Precup, Doina and Teh, Yee Whye},
  volume = 	 {70},
  series = 	 {Proceedings of Machine Learning Research},
  month = 	 {06--11 Aug},
  publisher =    {PMLR},
  url = 	 {https://proceedings.mlr.press/v70/avron17a.html}
}

@misc{shahrampour2019samplingrandomfeaturesempirical,
      title={On Sampling Random Features From Empirical Leverage Scores: Implementation and Theoretical Guarantees}, 
      author={Shahin Shahrampour and Soheil Kolouri},
      year={2019},
      eprint={1903.08329},
      archivePrefix={arXiv},
      primaryClass={cs.LG},
      url={https://arxiv.org/abs/1903.08329}, 
}

@InProceedings{pmlr-v97-li19k,
  title = 	 {Towards a Unified Analysis of Random {F}ourier Features},
  author =       {Li, Zhu and Ton, Jean-Francois and Oglic, Dino and Sejdinovic, Dino},
  booktitle = 	 {Proceedings of the 36th International Conference on Machine Learning},
  pages = 	 {3905--3914},
  year = 	 {2019},
  editor = 	 {Chaudhuri, Kamalika and Salakhutdinov, Ruslan},
  volume = 	 {97},
  series = 	 {Proceedings of Machine Learning Research},
  month = 	 {09--15 Jun},
  publisher =    {PMLR},
  url = 	 {https://proceedings.mlr.press/v97/li19k.html}
}

@inproceedings{liu2020random,
  title={Random {Fourier} Features via Fast Surrogate Leverage Weighted Sampling},
  author={Liu, Fanghui and Huang, Xiaolin and Chen, Yudong and Yang, Jie and Suykens, Johan},
  booktitle={Proceedings of the AAAI Conference on Artificial Intelligence},
  volume={34},
  pages={4844--4851},
  year={2020},
  doi={10.1609/aaai.v34i04.5920}
}

@book{wittek2014,
  author={Peter Wittek},
  title={Quantum Machine Learning: What Quantum Computing Means to Data Mining},
  year={2014},
  publisher={Elsevier},
  url = {https://www.sciencedirect.com/book/9780128009536/quantum-machine-learning}

}

@article{10.1098/rspa.2017.0551,
    author = {Ciliberto, Carlo and Herbster, Mark and Ialongo, Alessandro Davide and Pontil, Massimiliano and Rocchetto, Andrea and Severini, Simone and Wossnig, Leonard},
    title = {Quantum machine learning: a classical perspective},
    journal = {Proceedings of the Royal Society A: Mathematical, Physical and Engineering Sciences},
    volume = {474},
    number = {2209},
    pages = {20170551},
    year = {2018},
    month = {01},
    issn = {1364-5021},
    doi = {10.1098/rspa.2017.0551},
    url = {https://doi.org/10.1098/rspa.2017.0551},
}

@misc{isogai2026winninglotteryticketsneural,
      title={Winning Lottery Tickets in Neural Networks via a Quantum-Inspired Classical Algorithm}, 
      author={Natsuto Isogai and Hayata Yamasaki and Sho Sonoda and Mio Murao},
      year={2026},
      eprint={2605.13979},
      archivePrefix={arXiv},
      primaryClass={quant-ph},
      url={https://arxiv.org/abs/2605.13979}, 
}

@article{PhysRevA.99.052331,
  title = {Quantum-assisted {Gaussian} process regression},
  author = {Zhao, Zhikuan and Fitzsimons, Jack K. and Fitzsimons, Joseph F.},
  journal = {Phys. Rev. A},
  volume = {99},
  issue = {5},
  pages = {052331},
  numpages = {6},
  year = {2019},
  month = {May},
  publisher = {American Physical Society},
  doi = {10.1103/PhysRevA.99.052331},
  url = {https://link.aps.org/doi/10.1103/PhysRevA.99.052331}
}

@misc{chia2018quantuminspiredsublinearclassicalalgorithms,
      title={Quantum-inspired sublinear classical algorithms for solving low-rank linear systems}, 
      author={Nai-Hui Chia and Han-Hsuan Lin and Chunhao Wang},
      year={2018},
      eprint={1811.04852},
      archivePrefix={arXiv},
      primaryClass={cs.DS},
      url={https://arxiv.org/abs/1811.04852}, 
}

@InProceedings{chia_et_al:LIPIcs:2020:13391,
  author =	{Nai-Hui Chia and Andr{\'a}s Gily{\'e}n and Han-Hsuan Lin and Seth Lloyd and Ewin Tang and Chunhao Wang},
  title =	{{Quantum-Inspired Algorithms for Solving Low-Rank Linear Equation Systems with Logarithmic Dependence on the Dimension}},
  booktitle =	{31st International Symposium on Algorithms and Computation (ISAAC 2020)},
  pages =	{47:1--47:17},
  series =	{Leibniz International Proceedings in Informatics (LIPIcs)},
  ISBN =	{978-3-95977-173-3},
  ISSN =	{1868-8969},
  year =	{2020},
  volume =	{181},
  editor =	{Yixin Cao and Siu-Wing Cheng and Minming Li},
  publisher =	{Schloss Dagstuhl--Leibniz-Zentrum f{\"u}r Informatik},
  address =	{Dagstuhl, Germany},
  URL =		{https://drops.dagstuhl.de/opus/volltexte/2020/13391},
  URN =		{urn:nbn:de:0030-drops-133916},
  doi =		{10.4230/LIPIcs.ISAAC.2020.47}
}

@article{Gilyen2022improvedquantum,
  doi = {10.22331/q-2022-06-30-754},
  url = {https://doi.org/10.22331/q-2022-06-30-754},
  title = {An improved quantum-inspired algorithm for linear regression},
  author = {Gily{\'{e}}n, Andr{\'{a}}s and Song, Zhao and Tang, Ewin},
  journal = {{Quantum}},
  issn = {2521-327X},
  publisher = {{Verein zur F{\"{o}}rderung des Open Access Publizierens in den Quantenwissenschaften}},
  volume = {6},
  pages = {754},
  month = jun,
  year = {2022}
}

@inbook{doi:10.1137/1.9781611977912.86,
author = {Ainesh Bakshi and Ewin Tang},
title = {An Improved Classical Singular Value Transformation for Quantum Machine Learning},
booktitle = {Proceedings of the 2024 Annual ACM-SIAM Symposium on Discrete Algorithms (SODA)},
chapter = {},
year = {2024},
publisher = {Society for Industrial and Applied Mathematics},
pages = {2398-2453},
doi = {10.1137/1.9781611977912.86},
URL = {https://epubs.siam.org/doi/abs/10.1137/1.9781611977912.86},
eprint = {https://epubs.siam.org/doi/pdf/10.1137/1.9781611977912.86}
}

@article{Martyn2021grand,
  title = {Grand Unification of Quantum Algorithms},
  author = {Martyn, John M. and Rossi, Zane M. and Tan, Andrew K. and Chuang, Isaac L.},
  journal = {PRX Quantum},
  volume = {2},
  number = {4},
  pages = {040203},
  year = {2021},
  doi = {10.1103/PRXQuantum.2.040203},
  url = {https://doi.org/10.1103/PRXQuantum.2.040203}
}

\clearpage
\appendix
\section{Preliminaries}\label{sec: preliminaries}
\subsection{Notations and general definitions}
For a natural number $n$, we use $[n]$ to denote the set $\{1,\ldots,n\}$.
Let $\mb{R}$ and $\mb{C}$ denote the fields of real and complex numbers, respectively.
For any entry $v$ in $\mb{C}$, we use $|v|$ and $\overline{v}$ to denote the absolute value and the complex conjugate of $v$, respectively.
For a set $S$, let $\card(S)$ denote the cardinality of $S$.

Let $\mb{R}^n$ and $\mb{C}^n$ be $n$-dimensional vector spaces over $\mb{R}$ and $\mb{C}$, respectively.
For a vector $\bm{x} \in \mb{C}^n$, $\|{\bm{x}}\|_2$ and $\bm{x}_i$ denote the $l_2$-norm and the $i$-th entry of $\bm{x}$, respectively.
For a matrix $\bm{A} \in \mb{C}^{m\times n}$, $\bm{A}^\dagger, \|\bm{A}\|_{\mathrm{op}},\|\bm{A}\|_{\mathrm{F}}$ denote the conjugate transpose, operator norm, and Frobenius norm of $\bm{A}$, respectively.
For a Hermitian matrix $\bm{A} \in \mathbb{C}^{n \times n}$, $\lambda_{\max}(\bm{A})$ and $\lambda_{\min}(\bm{A})$ denote the maximum and minimum eigenvalues of $\bm{A}$, respectively, and $\sigma(\bm{A})$ denotes its spectrum, i.e., the set of eigenvalues of $\bm{A}$.
We denote the entry in $A$ at the $i$-th row and $j$-th column by $A(i,j)$.
Equivalently, with respect to the standard computational basis $\{\ket{i}\}_{i = 1}^n$, we may write $\bm{A}(i,j) = \bra{i}\bm{A} \ket{j}$.
Let $\bm{A}(i,\ast)$ and $\bm{A}(\ast, j)$ denote the $i$-th row and $j$-th column of $\bm{A}$.
When a matrix $\bm{A} \in \mathbb{C}^{n \times n}$ is diagonal, we sometimes write $\bm{A} = \diag(a_1,a_2,\ldots , a_n)$, and $a_i$ is the $i$-th diagonal entry of $\bm{A}$ for each $i$.
For a square matrix $\bm{A} \in \mathbb{C}^{n \times n}$, we write $\bm{A}_S$ for the $|S| \times |S|$ principal submatrix by restricting $\bm{A}$ to the rows and columns indexed by $S$, i.e., $\bm{A}_S(i,j) = \bm{A}(i,j)$ for all $i,j \in S$.
We say that a matrix $\bm{A} \in \mathbb{C}^{m \times n}$ is row $\alpha_r$-sparse if the number of non-zero entries in each row of $\bm{A}$ is at most $\alpha_r$.
Similarly, we say that a matrix $\bm{A} \in \mathbb{C}^{m \times n}$ is column $\alpha_c$-sparse if the number of non-zero entries in each column of $\bm{A}$ is at most $\alpha_c$.
We say that a matrix $\bm{A} \in \mathbb{C}^{m \times n}$ is $\alpha$-sparse if $\bm{A}$ is row $\alpha$-sparse and column $\alpha$-sparse.

We say that an algorithm is a randomized classical algorithm if the classical algorithm has access to random bits.
In this paper, arithmetic (i.e., addition and multiplication) for real values and function evaluation (e.g., evaluating $f(x)$ from an input $x$) are performed in $O(1)$ time.

In this paper, we use the total variation distance to measure the distance between two probability distributions.
\begin{definition}[Total variation distance]
    The total variation distance $d_\mathrm{TV}(\mc{D}, \mc{D}')$ between two distributions $\mc{D}$ and $\mc{D}'$ over a finite set $\mc{X}$ is defined as
    \begin{equation}
        d_\mathrm{TV}(\mc{D}, \mc{D}') = \frac{1}{2}\sum_{x\in \mc{X}} |\mc{D}(x) - \mc{D}'(x)|.
    \end{equation}
\end{definition}
It is easy to see that the total variation distance is a metric, i.e., it satisfies the triangle inequality, and $d_\mathrm{TV}(\mc{D}, \mc{D}') = 0$ if and only if $\mc{D} = \mc{D}'$, and $d_\mathrm{TV}(\mc{D}, \mc{D}') \leq 1$ for any pair of $\mc{D}$ and $\mc{D}'$.

We also use the following concentration inequality to bound the estimation error and failure probability.
\begin{lemma}[Hoeffding's inequality\label{lemma: Hoeffding's inequality}]
    Let $Z_1, Z_2, \ldots, Z_n$ be independent random variables such that $Z_i \in [a, b]$ for every $i = 1, 2, \ldots, n$.
    Then, for every $\epsilon > 0$, the following inequalities hold:
    \begin{align}
        \pr\left(\frac{1}{n}\sum_{i=1}^n (Z_i - \mathbb{E}[Z_i]) \geq \epsilon\right) &\leq \exp\left(-\frac{2n\epsilon^2}{(b - a)^2}\right),\\
        \pr\left(\frac{1}{n}\sum_{i=1}^n (Z_i - \mathbb{E}[Z_i]) \leq -\epsilon\right) &\leq \exp\left(-\frac{2n\epsilon^2}{(b - a)^2}\right), \label{eq: Hoeffding's inequality lower bound}\\
        \pr\left(\left|\frac{1}{n}\sum_{i=1}^n (Z_i - \mathbb{E}[Z_i])\right| \geq \epsilon\right) &\leq 2\exp\left(-\frac{2n\epsilon^2}{(b - a)^2}\right).
    \end{align}
\end{lemma}

\subsection{Matrix functions}
Let $\bm{A} \in \mathbb{C}^{m \times n}$ be a matrix, and $\bm{A} = \bm U \bm\Sigma \bm V^\dagger$ be a singular value decomposition (SVD), where $\bm U \in \mathbb{C}^{m \times m}$ and $\bm V \in \mathbb{C}^{n \times n}$ are unitaries, and $\bm \Sigma = \begin{bmatrix}
    \bm D & 0\\
    0 & 0
\end{bmatrix}$ is a rectangular diagonal matrix with a diagonal matrix $\bm D =\diag(d_1, \ldots, d_{r})\in \mathbb{R}^{r \times r}$ whose diagonal entries are nonzero singular values of $\bm A$.
\begin{definition}
    For a function $f:[0,\infty) \to \mb{C}$ and a matrix $\bm A\in \mathbb{C}^{m \times n}$, the singular value transformation of $\bm A$ is defined as
    \begin{equation}
        f^{\mathrm{SV}}(\bm A) = \bm U 
        \begin{bmatrix}
            f^{\mathrm{SV}}(\bm D) & 0 \\
            0 & 0
        \end{bmatrix} 
        \bm V^\dagger,
    \end{equation}
    where $f^{\mathrm{SV}}(\bm D) = \diag(f(d_1), \ldots, f(d_{r}))$, and the zero blocks are chosen so that the middle matrix has size $m \times n$.
\end{definition}

When $\bm A$ is Hermitian and $\bm A = \bm V \bm \Lambda \bm V^\dagger$ is an eigenvalue decomposition, where $\bm V$ is a unitary matrix, and $\bm \Lambda$ is a real diagonal matrix whose diagonal entries are eigenvalues of $\bm A$, the matrix function is defined similarly by using the eigenvalue decomposition of $\bm A$.
In particular, when $\bm A$ is a positive semidefinite matrix, the eigenvalues are equal to the singular values, and thus, the eigenvalue transformation of positive semidefinite matrices corresponds to the singular value transformation.

\begin{definition}
    For a function $f:\mb{R} \to \mb{C}$ and a Hermitian matrix $\bm A \in \mathbb{C}^{n \times n}$, the eigenvalue transformation of $\bm A$ is defined as
    \begin{equation}
        f^{\mathrm{EV}}(\bm A) = \bm V f^{\mathrm{EV}}(\bm \Lambda) \bm V^\dagger,
    \end{equation}
    where $f^{\mathrm{EV}}(\bm \Lambda) = \diag(f(\lambda_1),  \ldots, f(\lambda_{n}))$, and $\lambda_i$ is the $i$-th eigenvalue of $\bm A$ for each $i$.
\end{definition}
For simplicity of notation, we sometimes write $f(\bm A)$ instead of $f^{\mathrm{EV}}(\bm A)$ and $f^{\mathrm{SV}}(\bm A)$ when there is no confusion.

\subsection{Lipschitz constant}
\begin{definition}[Lipschitz function\label{def: Lipschitz function}]
    For a function $f:\mb{R} \to \mb{C}$, we say that $f$ is $L$-Lipschitz on $\mathfrak{F} \subseteq{\mb{R}}$, if there exists a constant $L >0$ such that for any $x,y \in \mathfrak{F}$,
    \begin{equation}
       |f(x) - f(y)| \leq  L|x - y|.
    \end{equation}
\end{definition}

\begin{lemma}[Translation preserves Lipschitz continuity\label{lem: parallel lipschitz}]
    For a constant $\gamma \in \mathbb{R}$ and an $L$-Lipschitz function $f:\mb{R} \to \mb{C}$ on $\mathfrak{F} \subseteq{\mb{R}}$, define $g: \mb{R} \to \mb{C}$ by $g(x) = f(\gamma + x)$ for all $x \in \mathbb{R}$.
    Then the function $g$ is $L$-Lipschitz on $\mathfrak{F}_\gamma \coloneq \{x \mid x + \gamma \in \mathfrak{F}\}$.
\end{lemma}
\begin{proof}
    For all $x,y \in \mathfrak{F}_\gamma$, it holds that
    \begin{equation}
        \begin{aligned}
            |g(x) - g(y)| &= |f(\gamma + x) - f(\gamma + y)| \\
            &\leq L |(\gamma + x) - (\gamma + y)| \\
            & = L |x - y|,
        \end{aligned}
    \end{equation}
    where the inequality follows from Definition~\ref{def: Lipschitz function}.
\end{proof}

\begin{lemma}[Operator-norm Lipschitz bound with a logarithmic factor~\cite{ALEKSANDROV20112741}\label{lem: operator-norm Lipschitz}]
    Let $\bm A, \bm B$ be two Hermitian matrices. Suppose that $f:\mb{R} \to \mb{C}$ is an $L$-Lipschitz function on $\sigma(\bm A) \cup \sigma(\bm B) \subseteq \mathbb{R}$.
    Then, there exists a numerical constant $C > 0$ such that
    \begin{equation}
        \ope{f(\bm A) - f(\bm B)} \leq C(1 + \log(m_{\bm{A},\bm{B}})) L \ope{\bm A - \bm B},
    \end{equation}
    where
    \begin{equation}
        m_{\bm{A},\bm{B}} \coloneq \min\{\card(\sigma(\bm A)), \card(\sigma(\bm B))\}
    \end{equation}
\end{lemma}

\subsection{Access models}
We introduce the definition of sampling-and-query access oracles, which is widely used in block-encoding and dequantization techniques~\cite{chia2022sampling,gilyen2019quantum,le2025robust,gharibian2022dequantizing}.

\begin{definition}[Query access~\cite{chia2022sampling}]
    For a vector $\bm v \in \mb{C}^n$, let $\mathrm{Q}(\bm v)$ be an oracle that for any $i \in [n]$, takes as input $i \in [n]$ and outputs the value of $\bm v(i)$.
    Similarly, for a matrix $\bm A \in \mb{C}^{m \times n}$, let $\mathrm{Q}(\bm A)$ be an oracle that for any $(i,j) \in [m]\times [n]$, takes as input $(i,j)$ and outputs the value of $\bm A(i,j)$.
    Let $\mathbf{q}(\bm v)$ and $\mathbf{q}(\bm A)$ be the time complexities of one query to the oracles $\mathrm{Q}(\bm v)$ and $\mathrm{Q}(\bm A)$, respectively.
\end{definition}

We also consider a noisy version of the query access oracle, which is an oracle that outputs a value close to the value of the entry of a vector or a matrix.
\begin{definition}[Noisy query access]
    For a vector $\bm v \in \mb{C}^n$ and $\delta \in [0,1]$, let $\mathrm{Q}_{\epsilon,\delta}(\bm v)$ be an oracle that for any $i \in [n]$, takes as input $i$ and outputs a value $\tilde{\bm v}(i)$ such that $|\bm v(i) - \tilde{\bm v}(i)| \leq \epsilon$ with probability at least $1-\delta$ for each query.
    Similarly, for a matrix $\bm A \in \mb{C}^{m \times n}$ and $\delta \in [0,1]$, let $\mathrm{Q}_{\epsilon,\delta}(\bm A)$ be an oracle that for any $(i,j) \in [m]\times [n]$, takes as input $(i,j)$ and outputs a value $\tilde{\bm A}(i,j)$ such that $|\bm A(i,j) - \tilde{\bm A}(i,j)| \leq \epsilon$ with probability at least $1 - \delta$.
    Let $\mathbf{q}_{\epsilon,\delta}(\bm v)$ and $\mathbf{q}_{\epsilon,\delta}(\bm A)$ be the time complexities to query the oracles $\mathrm{Q}_{\epsilon,\delta}(\bm v)$ and $\mathrm{Q}_{\epsilon,\delta}(\bm A)$, respectively.
\end{definition}

\begin{definition}[Sampling and query access to a vector~\cite{chia2022sampling}]
    For a vector $\bm v \in \mb{C}^n$, let $\mathrm{SQ}(\bm v)$ be an oracle that performs the following operations:
    \begin{itemize}
        \item For any $i \in [n]$, takes as input $i$ and outputs the value of $\bm v(i)$.
        \item Outputs an index $i \in [n]$ with probability $\frac{|\bm v(i)|^2}{\|\bm v\|_2^2}$.
        \item Outputs the value of $\|\bm v\|_2$.
    \end{itemize}
    Let $\mathbf{q}(\bm v),\mathbf{s}(\bm v)$, and $\mathbf{n}(\bm v)$ be the time complexities of querying entries, sampling indices, and computing the norm, respectively.
    Let $\mathbf{sq}(\bm v) \coloneq \max(\mathbf{q}(\bm v), \mathbf{s}(\bm v), \mathbf{n}(\bm v))$.
\end{definition}

Similarly, we can define the sampling-and-query access to a matrix by using the sampling-and-query access to vectors.
\begin{definition}[Sampling and query access to a matrix~\cite{chia2022sampling}]
    For a matrix $\bm A \in \mb{C}^{m \times n}$ and a vector $\bm a \in \mathbb{R}^m$ such that $\bm a(i) \coloneq \|A(i,\ast)\|_2$, let $\mathrm{SQ}(\bm A)$ be an oracle that performs the following operations:
    \begin{itemize}
        \item For any $i \in [m]$, acts as $\mathrm{SQ}(\bm A(i,\ast))$,
        \item Acts as $\mathrm{SQ}(\bm a)$,
    \end{itemize}
    The time complexities of querying entries, sampling rows, and computing row norms are denoted by $\mathbf{q}(\bm A) \coloneq \max(\mathbf{q}(\bm A(i,\ast)),\mathbf{q}(a))$, $\mathbf{s}(\bm A) = \max(\mathbf{s}(\bm A(i,\ast)),\mathbf{s}(\bm a))$, and $\mathbf{n}(\bm A) = \mathbf{n}(\bm a)$, respectively.
    Let $\mathbf{sq}(\bm A) \coloneq \max(\mathbf{q}(\bm A), \mathbf{s}(\bm A), \mathbf{n}(\bm A))$.
\end{definition}

We finally introduce the definition of sparse query access to a sparse matrix, which is an oracle that allows us to query an index in the support of each row and column of a matrix, and queries the value of the entry at that index.
The quantum version of this oracle is used to implement the block-encoding of a sparse matrix~\cite{PhysRevLett.103.150502}.
\begin{definition}[Sparse access to a matrix]
    For a matrix $\bm A \in \mb{C}^{m \times n}$ that is row $\alpha_r$-sparse and column $\alpha_c$-sparse, let $\mathrm{Q}_{\mathrm{sp}}(\bm A)$ be an oracle that performs the following operations:
    \begin{itemize}
        \item For any $i \in [m]$ and $l \in [\alpha_r]$, takes as input $(i,l)$ and outputs the column index of the $l$-th non-zero entry in the $i$-th row of $\bm A$, and $0$ if the number of non-zero entries in the $i$-th row is less than $l$,
        \item For any $j \in [n]$ and $l \in [\alpha_c]$, takes as input $(j,l)$ and outputs the row index of the $l$-th non-zero entry in the $j$-th column of $\bm A$, and $0$ if the number of non-zero entries in the $j$-th column is less than $l$,
        \item For any $i \in [m]$ and $j \in [n]$, acts as $\mathrm{Q}(\bm A)$ to output the value of $\bm A(i,j)$,
    \end{itemize}
    where $0$ is a sentinel value not in $[m]$ and $[n]$.
    The time complexity to query the oracle $\mathrm{Q}_{\mathrm{sp}}(\bm A)$ is denoted by $\mathbf{q}_{\mathrm{sp}}(\bm A)$.
\end{definition}

\section{Heavy-coordinate truncation of diagonal data}
In this appendix, we prove the basic truncation lemma used in Appendix~\ref{sec: main theorems}.
The setting is a nonnegative diagonal matrix $\bm{D}$ whose sampling-and-query access is given.
The key observation is that sampling-and-query access to a diagonal matrix samples indices with probability proportional to the squared diagonal entries.
Thus, repeated sampling detects all sufficiently large diagonal entries with high probability.

\begin{lemma}
    \label{lem: low rank approximation of D}
    Let $\bm{D} = \diag(d_1, \ldots, d_n) \in \mathbb{R}^{n \times n}$ be a nonnegative diagonal matrix.
    Given access to $\mathrm{SQ}(\bm{D})$, a truncation threshold $\epsilon > 0$, and a failure probability $\delta \in (0,1)$, there exists a randomized classical algorithm that outputs a classical data structure implementing $\sq(\widetilde{\bm D})$ for a nonnegative diagonal matrix $\widetilde{\bm{D}} = \diag(\tilde{d}_1, \ldots, \tilde{d}_n)$ such that $\supp(\widetilde{\bm{D}}) \subseteq \supp(\bm{D})$, $\rk{\widetilde{\bm{D}}} \leq \frac{\flo{\bm{D}}^2}{\epsilon^2}$, and $\ope{\bm{D} - \widetilde{\bm{D}}} \leq \epsilon $ with probability at least $1 - \delta$, with $\mathbf{sq}({\widetilde{\bm D}}) = O\left(\log(1 + \frac{\flo{\bm{D}}}{\epsilon})\right)$.
    The runtime is 
    \begin{equation}
        \tO\left((\mathbf{sq}(\bm{D}) + 1) \times \frac{\flo{\bm{D}}^2}{\epsilon^2}\log\left(\frac{1}{\delta}\right)\right).
    \end{equation}
\end{lemma}
\begin{proof}
    Initialize $\widetilde{\bm D} = 0$.
    If $\flo{\bm{D}} = 0$, then $\bm D = 0$, and the algorithm returns $\widetilde{\bm{D}}$.
    If $\epsilon \geq \flo{\bm{D}}$, then $\ope{\bm D} \leq \flo{\bm{D}} \leq \epsilon$, and returning $\widetilde{\bm D}$ already satisfies the desired operator-norm bound.
    If the returned matrix $\widetilde{\bm D}$ is zero, the query and norm operations are defined in the obvious way, and the sampling operation may return an arbitrary fixed index.
    
    Therefore, in the following, we assume that $\flo{\bm{D}} > 0$ and $\epsilon < \flo{\bm{D}}$.
    The classical algorithm is as follows:
    \begin{enumerate}
        \item Set $R \coloneq \left\lceil \frac{\flo{\bm D}^2}{\epsilon^2}\ln\left(\frac{\flo{\bm D}^2}{\epsilon^2\delta}\right)\right\rceil$, and $F = 0$, and let $S$ be an empty set.
        \item For $r = 1$ to $R$, sample $i_r \in [n]$ with a probability $d_{i_r}^2/\flo{D}^2$ by using the oracle $\mathrm{SQ}(\bm D)$.
        \item If $i_r$ has not already been inserted into $S$ and $d_{i_r} \geq \epsilon$, then insert $i_r$ into $S$ and set $\tilde{d}_{i_r} \coloneq d_{i_r}$ and $F \mathrel{+}= d_{i_r}^2$ by using the oracle $\sq(\bm D)$.
        \item End for.
        \item $F = \sqrt{F}$.
    \end{enumerate}

    The correctness of the algorithm follows from the fact that $\bm D$ and $\widetilde{\bm D}$ are nonnegative diagonal matrices.
    The operator norm of the difference between the two matrices is bounded by $\ope{\bm D - \widetilde{\bm D}} \leq \max_{i \in [n]}\{|d_i - \tilde{d}_i|\}$.
    Thus, when every index $i$ such that $d_i \geq \epsilon$ is sampled, then $\|\bm D-\widetilde{\bm D}\|_{\mathrm{op}} = \max_{i \in [n]}\{|d_i - \tilde{d}_i|\} \leq  \epsilon$.

    For an index $i \in [n]$ such that $d_i \geq \epsilon$, the probability that the index $i$ is not sampled in the $R$ trials is at most $(1 - d_i^2/\flo{\bm D}^2)^R \leq (1 - \epsilon^2/\flo{\bm D}^2)^R \leq e^{-R\epsilon^2/\flo{\bm D}^2}$.
    And, the number of indices $i \in [n]$ such that $d_i \geq \epsilon$ is at most $\lfloor\frac{\flo{\bm D}^2}{\epsilon^2}\rfloor \leq \frac{\flo{\bm D}^2}{\epsilon^2}$, which corresponds to the rank of the matrix $\widetilde{\bm D}$, i.e., $\rk{\widetilde{\bm D}} \leq \frac{\flo{\bm D}^2}{\epsilon^2}$.
    Therefore, the probability that there exists an index $i \in [n]$ such that $d_i \geq \epsilon$ and is not sampled in the $R$ trials is bounded by the union bound as follows:
    \begin{align}
    \frac{\flo{\bm D}^2}{\epsilon^2} e^{-R\frac{\epsilon^2}{\flo{\bm D}^2}} 
    & \leq \frac{\flo{\bm D}^2}{\epsilon^2} e^{-\frac{\flo{\bm D}^2}{\epsilon^2}\ln\left(\frac{\flo{\bm D}^2}{\epsilon^2\delta}\right)\frac{\epsilon^2}{\flo{\bm D}^2}} \\ \notag
    & \leq \delta.
    \end{align}
    The correctness of the value of the Frobenius norm of the matrix $\widetilde{\bm D}$ is trivial, i.e., $F = \|\widetilde{\bm D}\|_F$.
    
    To implement $\sq(\widetilde{\bm D})$, we use the standard $\ell_2$-sampling data structure used in quantum-inspired algorithms and QRAM-based input models~\cite{kerenidis_et_al:LIPIcs.ITCS.2017.49, 10.1145/3313276.3316310,chia2022sampling}.
    Given the list of nonzero entries $\{i,\tilde d_i: i \in S\}$, this data structure can be constructed in $O(|S| \log(|S|))$ time and $O(|S|)$ space, and supports entry queries, $\ell_2$-sampling, and norm queries with
    \begin{equation}
    \begin{aligned}
        \mathbf{q}(\widetilde{\bm D}) = O(\log(1 + |S|)), \\
        \mathbf{s}(\widetilde{\bm D}) = O(\log(1 + |S|)), \\
        \mathbf{n}(\widetilde{\bm D}) = O(1).
    \end{aligned}
    \end{equation}
    For a diagonal matrix, this vector structure for the diagonal entries implements $\sq(\widetilde{\bm D})$ with the same cost $\mathbf{sq}({\widetilde{\bm D}}) = O\left(\log(1 + \frac{\flo{\bm{D}}}{\epsilon})\right)$.

    The runtime of this algorithm is dominated by the sum of the runtime of sampling $R$ indices from the distribution $\bm D$ in $O(R \times \mathbf{sq}(\bm D))$ time, the runtime of checking whether $d_{i} \geq \epsilon$ for each sampled index $i$ in $O(R \times \mathbf{sq}(\bm D))$ time, the runtime of computing the Frobenius norm $\flo{\widetilde{\bm D}}$ in $O(R \times \mathbf{sq}(\bm D))$ time, and the runtime to construct $\sq(\widetilde{\bm D})$ in $O(|S| \log(|S|))$.
    Therefore, the total time complexity of the algorithm is $O( R \times (\mathbf{sq}(\bm D) + \log(R))) = \tO\left((\mathbf{sq}(\bm{D}) + 1) \frac{\flo{\bm{D}}^2}{\epsilon^2}\log\left(\frac{1}{\delta}\right)\right)$, which yields the conclusion.
\end{proof}

\section{Diagonal truncation method}\label{sec: main theorems}
This appendix provides the main technical result of the paper.
We show that sampling-and-query access to the diagonal matrix $\bm D$, together with query access to the Hermitian matrix $\bm A$, is sufficient to construct a sparse classical representation of the matrix function $f(\gamma \bm{I} + \bm{DAD})$.
The algorithm does not exactly construct sampling-and-query access to the composite matrix $\bm{DAD}$.
Instead, it first applies Lemma~\ref{lem: low rank approximation of D} to replace $\bm D$ by a low-support diagonal matrix $\widetilde{\bm{D}}$, and then, computes the matrix function on the resulting small block.

\begin{algorithm}[tbp]
\SetAlgoLined
\caption{Diagonal truncation and transformation~\label{alg: type1}}
\KwIn{Precision parameters $0<\delta_D,\delta_A < 1$, $\epsilon_D,\epsilon_A > 0$, $\gamma \in \mathbb{R}$, access to $\mathrm{Q}_{\epsilon_A,\delta_A}(\bm A),\mathrm{SQ}(\bm D)$, the norms $\ope{\bm A},\ope{\bm D}$, and an evaluation oracle of $f: \mathbb{R} \to \mathbb{R}$ with the Lipschitz constant $L$.}
\KwOut{A data structure $(S,J_1,\ldots,J_{|S|},\pi,\bm B,c)$ implementing support-sparse query access to $\bm N$ in Theorem~\ref{thm: type1}, i.e., a support set $S \subseteq [n]$, lists $J_1,\ldots,J_{|S|} \subseteq{S}$, a bijection $\pi$, a $|S| \times |S|$ matrix $\bm B$, a constant $c$.}

Construct a matrix $\widetilde{\bm D}$ using $\mathrm{SQ}({\bm D})$ and Lemma~\ref{lem: low rank approximation of D} with threshold $\epsilon_D$ and failure probability $\delta_D$, and let $S = \supp(\widetilde{\bm D})$ and $s = |S|$\;
Create an ordering $\{i_1, i_2, \ldots, i_s\}$ of $S$ and define the bijection $\pi:S \to [s]$ with $\pi(i_t) = t$\;

Construct a Hermitian matrix $\widetilde{\bm M}_S \in \mathbb{C}^{s \times s}$ by querying the upper triangular entries and setting, for $t,u \in [s]$,
\begin{equation}
    \widetilde{\bm M}_S(t,u) \coloneq \gamma \delta_{t,u} + \tilde{d}_{i_t}\mathrm{Q}_{\epsilon_A,\delta_A}(\bm A)(i_t,i_u)\tilde{d}_{i_u}
\end{equation}
and $\widetilde{\bm M}_S(t,u) = \overline{\widetilde{\bm M}_S(u,t)}$\;
Compute an eigenvalue decomposition $\widetilde{\bm M}_S = \bm V \bm \Lambda \bm V^{\dagger}$\;
Compute $\bm B \coloneq \bm V\,f(\bm \Lambda)\,\bm V^\dagger$ by applying $f$ entrywise to the diagonal entries of $\bm \Lambda$\;
Make lists $J_t = \{i_u\in S: \bm B(t,\pi(i_u)) \neq 0\}$ for each $t \in [s]$\;
Compute $f(\gamma)$\;
\Return $(S,J_1,\ldots,J_{|S|},\pi,\bm B,f(\gamma))$.
\end{algorithm}

\begin{theorem}[Diagonal truncation approximation]\label{thm: type1}
    Let $\gamma \in \mathbb{R}$, $\bm{A} \in \mathbb{C}^{n \times n}$ be a Hermitian matrix, and $\bm D = \diag(d_1, d_2,\ldots , d_n) \in \mathbb{R}^{n\times n}$ be a nonnegative diagonal matrix.
    Define $\bm M \coloneq \gamma \bm{I} + \bm{DAD}$.
    Suppose that $f: \mathbb{R} \to \mathbb{R}$ is $L$-Lipschitz on $[\lambda_{\min}(\bm{M}) - \nu, \lambda_{\max}(\bm{M}) + \nu] \cup \{\gamma\}$ for some $\nu \geq 2\ope{\bm A}\ope{\bm D}\epsilon_D + \epsilon_A \flo{\bm D}^2$.
    Given precision parameters $\epsilon_D, \epsilon_A > 0$, failure probabilities $0 < \delta_D,\delta_A < 1$, access to $\mathrm{SQ}(\bm D)$ and $\mathrm{Q}_{\epsilon_A,\delta_A}(\bm A)$, the norms $\ope{\bm A},\ope{\bm D}$, and an evaluation oracle of $f$, Algorithm~\ref{alg: type1} outputs a data structure implementing sparse query access $\mathrm{Q}_{\mathrm{sp}}(\bm N)$ of $\mathbf{q}_{\mathrm{sp}}(\bm N) = O(1 + \log(\flo{\bm D}/\epsilon_D))$ of the following form.

    For a set $S \subseteq \supp(D)$ with $s = |S|$, a bijection $\pi: S\to [s]$, and a matrix $\bm{B} \in \mathbb{C}^{s \times s}$,
    \begin{equation}
        \bm N(i,j) \coloneq
            \begin{cases}
            \bm B(\pi(i),\pi(j)) & (i,j \in S),\\
            f(\gamma) & (i=j \notin S),\\
            0 & \text{otherwise}.
            \end{cases}
    \end{equation}
    With probability at least
    \begin{equation}
        1 - (\delta_D + \frac{s(s+1)}{2}\delta_A),
    \end{equation}
    the output satisfies
    \begin{align}\label{eq: the precision of type1}
        \ope{f(\bm{M}) - \bm N} &\leq C\left(1 + \log\left(1 + s \right)\right)L \\ \notag
        &\quad \times \left( 2\ope{\bm A}\ope{\bm D}\epsilon_D + \epsilon_A \flo{\bm D}^2 \right),
    \end{align}
    where $s \leq \left({\flo{\bm D}}/{\epsilon_D}\right)^2$ and $C$ is a numerical constant in Lemma~\ref{lem: operator-norm Lipschitz}.
    The runtime is
    \begin{equation}
        O(\mathbf{sq}(\bm D) + \mathbf{q}_{\epsilon_A,\delta_A}(\bm A) + 1) \times \tO\left(\left(\frac{\flo{\bm D}}{\epsilon_D}\right)^6\log\left(\frac{1}{\delta_D}\right)\right).
    \end{equation}
\end{theorem}

\begin{proof}
    We first show the approximation guarantee and then, analyze the access structure and the running time.

    By using Lemma~\ref{lem: low rank approximation of D} with parameters $(\epsilon_D,\delta_D)$, Algorithm~\ref{alg: type1} firstly approximates $\bm D$ by $\widetilde{\bm D}$ with 
    \begin{align}\label{eq: approximate D}
        \ope{\bm D - \widetilde{\bm D}} &\leq \epsilon_D, \\
        \rk{\widetilde{\bm D}} &\leq \frac{\flo{\bm D}^2}{\epsilon_D^2}, \label{eq: type1 rank tildeD}
    \end{align}
    with a probability at least $1 - \delta_D$.
    Let
    \begin{equation}\label{eq: type1 s}
        S \coloneq \supp(\widetilde{\bm D}),\qquad s \coloneq |S| = \rk{\widetilde{\bm D}} \leq \flo{\bm D}^2/\epsilon_D^2.
    \end{equation}

    We next consider the error coming from truncation and from noisy queries to $\bm{A}$.
    Let $\widetilde{\bm A}$ denote the Hermitian matrix whose entries on $S \times S$ are the values returned by $\mathrm{Q}_{\epsilon_A, \delta_A}(\bm{A})$ and their Hermitian conjugates.
    Outside $S \times S$, the entries of $\widetilde{\bm A}$ can be chosen arbitrarily, since they vanish after left and right multiplication by $\widetilde{\bm D}$.

    Let
    \begin{equation}
        \Delta_D \coloneq \bm{D} - \widetilde{\bm{D}}, \quad \Delta_A \coloneq \bm{A} -\widetilde{\bm{A}}.
    \end{equation}
    Then, we have
    \begin{equation}\label{eq: DAD - tilde DAD}
        \begin{aligned}
            &\ope{\bm{DAD} - \bm{\widetilde{D}\widetilde{A}\widetilde{D}}} \\
            &\quad= \ope{\Delta_D \bm A\bm D + \widetilde{\bm D}\bm A\Delta_D + \widetilde{\bm D}\Delta_A\widetilde{\bm D}} \\
            &\quad \leq 2\ope{\bm A}\ope{\bm D}\ope{\Delta_D} + \epsilon_A \flo{\widetilde{\bm D}}^2\\
            &\quad \leq 2\ope{\bm A}\ope{\bm D}\epsilon_D + \epsilon_A \flo{\widetilde{\bm D}}^2 \\
            &\quad \leq \nu,
        \end{aligned}
    \end{equation}
    where we used $\ope{\widetilde{\bm{D}}} \leq \ope{ \bm D}$ and $\flo{\widetilde{\bm{D}}}\leq \flo{ \bm D}$.
    Define the truncated matrix
    \begin{equation}
        \widetilde{\bm M} \coloneq \gamma \bm{I} + \widetilde{\bm D} \widetilde{\bm A} \widetilde{\bm D}.
     \end{equation}
    Then, by Weyl's inequality, we have
    \begin{equation}
        \sigma(\widetilde{\bm{M}}) \subseteq [\lambda_{\min}(\bm{M}) - \nu, \lambda_{\max}(\bm{M}) + \nu] \cup \{\gamma\}.
    \end{equation}
    Thus, for a function $g: \mathbb{R} \to \mathbb{R}$ defined by $g(x) = f(\gamma + x)$ for all $x \in \mathbb{R}$, we have
    \begin{equation}
        \ope{f(\bm M) - f(\widetilde{\bm M})} = \ope{g(\bm{DAD}) - g(\widetilde{\bm D}\widetilde{\bm A}\widetilde{\bm D})}
    \end{equation}
    from Lemma~\ref{lem: parallel lipschitz}.
    Since $\widetilde{\bm D}\widetilde{\bm A}\widetilde{\bm D}$ is supported on $S$, we have $\card(\sigma(\widetilde{\bm D}\widetilde{\bm A}\widetilde{\bm D})) \leq 1+s$.
    Therefore, we obtain
    \begin{equation}
        \begin{aligned}
            &\ope{f(\bm M) - f(\widetilde{\bm M})} \\
            &\quad \leq C\left(1 + \log\left(1 + s \right)\right) L \ope{\bm{DAD} - \widetilde{\bm D}\widetilde{\bm A}\widetilde{\bm D}} \\
            &\quad \leq C\left(1 + \log\left(1 + s \right)\right) \\
            &\qquad \times L \left( 2\ope{\bm A}\ope{\bm D}\epsilon_D + \epsilon_A \flo{\widetilde{\bm D}}^2 \right),
            \end{aligned}
    \end{equation}
    where the first inequality follows from Lemma~\ref{lem: operator-norm Lipschitz}, and the second follows from Eq.~\eqref{eq: DAD - tilde DAD}.

    Next, we show that the output of Algorithm~\ref{alg: type1} implements support-sparse query access to $\bm N = f(\widetilde{\bm M})$ with a high probability.
    
    The matrix $\widetilde{\bm M}$ is block diagonal after permuting the indices in $S$ to the first $s$ indices, i.e.,
    \begin{equation}
        \widetilde{\bm{M}} = \bm{P}_S^\top \begin{bmatrix}
                \gamma \bm I_s + (\widetilde{\bm D}\widetilde{\bm A}\widetilde{\bm D})_S & 0 \\
                0 & \gamma I_{n-s} \end{bmatrix} \bm{P}_S,
    \end{equation}
    where $P_S$ is a permutation matrix that moves the indices in $S$ to the first $s$ positions and the remaining indices to the last $n-s$ positions.
    Therefore, by the definition of matrix functions,
    \begin{equation}\label{eq: block diagonal function}
        f(\widetilde{\bm{M}})=  \bm{P}_S^\top \begin{bmatrix}
                f\left(\gamma \bm I_s + (\widetilde{\bm D}\widetilde{\bm A}\widetilde{\bm D})_S \right) & 0 \\
                0 &  f(\gamma) \bm I_{n-s}
            \end{bmatrix} \bm{P}_S 
    \end{equation}
    Then, the matrix $\widetilde{\bm M}_S$ constructed in Algorithm~\ref{alg: type1} is equal to
    \begin{equation}
        \widetilde{\bm M}_S = \gamma \bm{I}_s + (\widetilde{\bm D}\widetilde{\bm A}\widetilde{\bm D})_S.
    \end{equation}
    Therefore, if $\widetilde{\bm M}_S = \bm{V} \bm{\Lambda} \bm{V}^\dagger$ is its eigenvalue decomposition, the matrix $\bm B$ is exactly
    \begin{equation}
        \bm B \coloneq \bm{V} f(\bm{\Lambda}) \bm{V}^\dagger = f\left(\gamma \bm I_s + (\widetilde{\bm D}\widetilde{\bm A}\widetilde{\bm D})_S \right).
    \end{equation}
    Eq.~\eqref{eq: block diagonal function} thus shows that the data structure returned by Algorithm~\ref{alg: type1} represents $\bm{N} = f(\widetilde{\bm M})$, whose entries are given as
    \begin{equation}
        \bm N(i,j) \coloneq
            \begin{cases}
            \bm B(\pi(i),\pi(j)) & (i,j \in S),\\
            f(\gamma) & (i=j \notin S),\\
            0 & \text{otherwise}.
            \end{cases}
    \end{equation}

    The success probability is obtained by combining the success probabilities of the truncation step and the noisy queries to $\bm{A}$.
    The truncation step succeeds with probability at least $1 - \delta_D$.
    Conditioned on the selected support $S$, the construction of the Hermitian matrix $\widetilde{\bm M}_S$ uses $s(s+1)/2$ queries to $\mathrm{Q}_{\epsilon_A,\delta_A}(\bm A)$.
    By the union bound, all these queries are accurate with probability at least
    \begin{equation}
        1 - \frac{s(s+1)}{2}\delta_A.
    \end{equation}
    Therefore, the overall success probability is at least
    \begin{equation}
        (1-\delta_D)\left( 1 - \frac{s(s+1)}{2}\delta_A \right) \geq 1 - (\delta_D + \frac{s(s+1)}{2}\delta_A).
    \end{equation}

    We now describe the sparse query access.
    The query oracle is implemented as follows:
        \begin{itemize}
        \item input $(i,j) \in [n]\times [n]$
        \item if $i \in S$ and $j \in S$, then output $B(\pi(i),\pi(j))$,
        \item else if $i = j$, then output $f(\gamma)$,
        \item else output $0$.
    \end{itemize}
    
    The row sparse access is implemented by $S$, $f(\gamma)$, and $J_1,\ldots,J_{s}$ which are defined as $J_t = \{i_u\in S: (Vf(\Lambda)V^{\dagger})(t,\pi(i_u)) \neq 0\}$ for each $t \in [s]$, as follows:
    \begin{itemize}
        \item input $i \in [n]$ and $l \in [s]$,
        \item if $i \in S$, output $J_{\pi(i)}[l]$ if $l \leq |J_{\pi(i)}|$, and else output $0$ if $l > |J_{\pi(i)}|$, 
        \item if $i \notin S$, output $i$ if $l = 1$ and $f(\gamma) \neq 0$, and $0$ otherwise,
    \end{itemize}
    where $J_t[l]$ is the $l$-th entry in the list $J_t$.
    The column sparse access is analogously implemented because $N$ is Hermitian whenever $f$ is a real-valued function on the relevant spectrum.
    These runtimes are dominated by the runtimes of checking whether $i\in S$ and computing the bijection $\pi$, which are $O(\log(s))$ time by constructing a binary search tree whose nodes store $(i,\pi(i))$ for $i \in S$.

    Finally, we analyze the running time.
    The runtime of Algorithm~\ref{alg: type1} is dominated by the runtime of constructing $\widetilde{\bm D}$ by using Lemma~\ref{lem: low rank approximation of D}, which is $\tO\left((\mathbf{sq}(\bm D) + 1) \frac{\flo{\bm D}^2}{\epsilon_D^2}\log\left(\frac{1}{\delta_D}\right) \right)$, the runtime of constructing $\widetilde{\bm M}_S$ by using $\mathrm{Q}_{\epsilon_A, \delta_A}(\bm A)$, which is $O(s^2 \mathbf{q}_{\epsilon_A,\delta_A}(\bm A))$, and the runtime of computing the eigenvalue decomposition of $\widetilde{\bm M}_S$ and computing $\bm V f(\bm \Lambda) \bm V^\dagger$, which is $O(s^3)$.
    Therefore, due to Eq.~\eqref{eq: type1 s}, the total runtime of Algorithm~\ref{alg: type1} is 
    \begin{equation}
        \begin{aligned}
            &\tO\left((\mathbf{sq}(\bm D) + 1) \frac{\flo{\bm D}^2}{\epsilon_D^2}\log\left(\frac{1}{\delta_D}\right) \right)\\
            &\qquad + O\left(s^2 \mathbf{q}_{\epsilon_A,\delta_A}(A) + s^3\right)\\
            &\quad = O(\mathbf{sq}(\bm D) + \mathbf{q}_{\epsilon_A,\delta_A}(\bm A) + 1) \\
            &\qquad \times \tO\left(\left(\frac{\flo{\bm D}}{\epsilon_D}\right)^6\log\left(\frac{1}{\delta_D}\right)\right),
        \end{aligned}
    \end{equation}
    which completes the proof of Theorem~\ref{thm: type1}.
\end{proof}

The factor $C(1+\log(1+s))L$ in Eq.~\eqref{eq: the precision of type1}
arises from a general operator-norm Lipschitz bound and therefore yields a
uniform error estimate for arbitrary Lipschitz functions $f$.
For special functions $f$ for which stronger operator-norm Lipschitz bounds are
available, the resulting approximation error can be bounded more sharply.

\begin{remark}[Relation to sampling and query access to the transformed matrix]
    Theorem~\ref{thm: type1} should not be interpreted as constructing sampling-and-query access to the full transformed matrix $\bm{DAD}$.
    Standard quantum-inspired dequantization methods often assume sampling access to the matrix to be transformed~\cite{chia2022sampling,10.1145/3313276.3316310}.
    On the other hand, in the present input model, we are given only query access $\mathrm{Q}(\bm A)$ to the Hermitian matrix $\bm{A}$.
    And thus, sampling access to $\bm{DAD}$ is not assumed.

    The key point of Theorem~\ref{thm: type1} is that such access is unnecessary for operator-norm approximation of $f\left(\gamma \bm I + \bm{DAD} \right)$.
    Algorithm~\ref{alg: type1} replaces $\bm D$ by a low-support diagonal matrix $\widetilde{\bm D}$ such that $\| \bm D - \widetilde{\bm D} \|_\mathrm{op} \leq \epsilon_D$ in~\eqref{eq: approximate D}.
    In fact, once the support $S = \supp(\widetilde{\bm D})$ and $\widetilde{\bm D}$ are known, one can also construct sampling-and-query access to the approximated matrix $\widetilde{\bm D} \bm A \widetilde{\bm D}$ by using only $\mathrm{Q}(\bm A)$, since the $s = |S|$ is sufficiently small as in~\eqref{eq: type1 s}.
    Thus, the sampling access to the true matrix $\bm{DAD}$ cannot be efficiently constructed only with $\mathrm{SQ}(\bm D)$ and $\mathrm{Q}(\bm A)$ in general, but we avoid this stronger requirement by constructing a support-sparse approximation to the matrix to be transformed.
\end{remark}

\section{Application to learning with optimized random features}\label{sec: optimized random features}
We now apply the diagonal-truncation theorem to the quantum algorithm for learning with optimized random features~\cite{10.5555/3495724.3496871}.
The relevant quantum subroutine applies a matrix inverse to a diagonally weighted kernel matrix, which is not covered by the previous dequantization techniques~\cite{10.1145/3313276.3316310,chia2018quantuminspiredsublinearclassicalalgorithms,chia_et_al:LIPIcs:2020:13391,jethwani_et_al:LIPIcs.MFCS.2020.53,Gilyen2022improvedquantum,chia2022sampling,le2025robust,doi:10.1137/1.9781611977912.86,gharibian2022dequantizing}.
Although this matrix to be transformed is not given with sampling-and-query access, its diagonal part is given through the oracle of an empirical distribution, and the kernel is specified through a kernel-description oracle.
This is precisely the access model addressed by Theorem~\ref{thm: type1}.

\subsection{Learning with optimized random features}
We consider supervised learning with training examples
$\{(x_i,y_i)\}_{i=1}^N \subset \mathbb{R}^D \times \mathbb{R}$ drawn i.i.d.\
from an unknown distribution by using the kernel method~\cite{scholkopf2002learning}.
The target is to learn a predictor by using random features for a shift-invariant kernel.
A random-feature method approximates a kernel function by mapping each input $x \in \mathbb{R}^D$ to a finite-dimensional feature vector whose entries are randomly sampled basis functions.
For shift-invariant kernels, this construction is based on the Fourier representation of the kernel, so that the predictor can be trained as a linear model in the random-feature space~\cite{NIPS2007_013a006f,NIPS2008_0efe3284}.
In learning with optimized random features, instead of sampling features from the original spectral measure of the kernel, one samples them from an importance-weighted probability distribution designed to reduce the number of required features while preserving the approximation quality of kernel ridge regression~\cite{JMLR:v18:15-178}.
See \cite{NIPS2007_013a006f,NIPS2008_0efe3284,10.5555/3495724.3496871} for more details.
In the paper~\cite{10.5555/3495724.3496871}, the authors show that learning with optimized random features can be done efficiently by using a quantum subroutine to prepare or sample optimized random features efficiently.

\subsection{The subroutine of learning with optimized random features}
The core subroutine in learning with optimized random features is inversion of a high-dimensional matrix by using QSVT techniques, which is not covered by the previous dequantization techniques~\cite{10.1145/3313276.3316310,chia2018quantuminspiredsublinearclassicalalgorithms,chia_et_al:LIPIcs:2020:13391,jethwani_et_al:LIPIcs.MFCS.2020.53,Gilyen2022improvedquantum,chia2022sampling,le2025robust,doi:10.1137/1.9781611977912.86,gharibian2022dequantizing}.

We summarize the setting of this subroutine, following the notation in~\cite{10.5555/3495724.3496871}.
In the paper~\cite{10.5555/3495724.3496871}, the authors consider the problem by discretizing the input domain of real-valued space $\mathbb{R}^D$ by $G \in \mathbb{N}$, and thus, the input domain is defined as $\mathcal{X}=\{0,1,\ldots,G-1\}^D$, and we let $n = |\mathcal{X}|$, i.e., $n = G^D$.
Let $\bm F_D \in \mathbb{C}^{n \times n}$ be the $D$-dimensional discrete Fourier transform on $\mathcal{X}$, i.e.,
\begin{equation}
    \bm F_D(x,y) \coloneq \frac{1}{\sqrt{G^D}} \exp\left(- \frac{2\pi i x \cdot y}{G}\right) \quad (x,y \in \mathcal{X}),
\end{equation}
where $x \cdot y$ is the inner product of $x$ and $y$ over $\mathbb{Z}_G$.
We are given a nonnegative function $Q^{(\tau)}:\mathcal{X}\to\mathbb{R}_{\ge 0}$ and an empirical probability distribution $\hat q^{(\rho)}:\mathcal{X}\to\mathbb{R}_{\ge 0}$, i.e., $\sum_{x \in \mathcal{X}} \hat q^{(\rho)}(x) = 1$.
By a slight abuse of notation, we also define diagonal matrices $\bm Q^{(\tau)}, \hat{\bm q}^{(\rho)} \in \mathbb{R}^{n \times n}$ as
\begin{equation}
\bm Q^{(\tau)}\coloneq \sum_{x\in \mathcal{X}} Q^{(\tau)}(x) \ketbra{x},
\qquad
\hat{\bm q}^{(\rho)} \coloneq \sum_{x\in \mathcal{X}} \hat q^{(\rho)}(x) \ketbra{x},
\end{equation}
and the associated kernel matrix
\begin{equation}\label{eq: k}
\bm k \coloneq \bm F_D^\dagger \bm Q^{(\tau)} \bm F_D,
\end{equation}
and its $(0,0)$-th entry $\bm k(0,0) = \Omega(1)$, where $\bm k$ is a positive semidefinite real symmetric matrix since $\bm k$ is defined by the kernel matrix (see~\cite{10.5555/3495724.3496871} for further details of the settings).

We are given query access $\mathrm{Q}(\bm Q^{(\tau)})$ and its operator norm 
\begin{equation}\label{eq: Q_max}
    Q^{(\tau)}_{\max} = \ope{\bm Q^{(\tau)}}
\end{equation}
of $\bm Q^{(\tau)}$.
Also, in the paper~\cite{10.5555/3495724.3496871}, the authors assume that a data structure implementing a quantum oracle $\mathcal{O}_\rho$ for preparing the quantum state is given,
\begin{equation}
    \mathcal{O}_\rho(\ket{0}) = \sum_{x\in\mathcal{X}} \sqrt{\hat q^{(\rho)}(x)}\ket{x},
\end{equation}
which is equivalent to having query access to $\mathrm{SQ}(\sqrt{\hat{\bm q}^{(\rho)}})$, so we assume that we are given access to $\mathrm{SQ}(\sqrt{\hat{\bm q}^{(\rho)}})$.
Then, we define the distribution $\mathcal{D}_\epsilon$ to be sampled as follows.
\begin{definition}[Optimized probability distribution\label{def: Optimized probability distribution}]
    For $\epsilon>0$, define the optimized probability distribution $\mathcal{D}_\epsilon$ over $\mathcal{X}$ by
    \begin{equation}
    \mathcal{D}_\epsilon(x)\propto
    \Braket{x \left|
    \sqrt{\bm Q^{(\tau)}}\,\bm F_D^\dagger \sqrt{\hat{\bm q}^{(\rho)}}
    \bm{\Sigma}_\epsilon^{-1}
    \sqrt{\hat{\bm q}^{(\rho)}}\,\bm F_D\sqrt{\bm Q^{(\tau)}}
    \right| x}.
    \end{equation}
    where $\bm{\Sigma}_\epsilon$ is defined by
    \begin{equation}
       \bm \Sigma_\epsilon \coloneq \epsilon \bm I + \sqrt{\hat{\bm q}^{(\rho)}}\,\bm k\,\sqrt{\hat{\bm q}^{(\rho)}}.
    \end{equation}
\end{definition}

The task of sampling optimized random features is to approximately sample from $\mathcal{D}_\epsilon$.
The quantum algorithm in~\cite{10.5555/3495724.3496871} achieves this sampling task within the error $\dtv(\mathcal{D}_\epsilon, \tilde{\mathcal{D}}) \leq \delta$ for $\delta \in (0,1]$ under the above access assumptions with a runtime
\begin{equation}\label{eq: quantum runtime}
    \begin{aligned}
        &O\left(D\log(G)\log\log(G) + \mathbf{q}(\bm Q^{(\tau)}) + \mathbf{sq}\left(\sqrt{\hat{\bm q}^{(\rho)}}\right)\right) \\
        &\quad \times \tO\left(\frac{Q^{(\tau)}_{\max}}{\epsilon}\polylog\left(\frac{1}{\delta}\right)\right).
    \end{aligned}
\end{equation}

\subsection{Dequantization of sampling from the optimized probability distribution}
By using Theorem~\ref{thm: type1}, we construct a randomized classical algorithm for sampling from $\mathcal{D}_\epsilon$ with a runtime polynomially related to the quantum algorithm in~\cite{10.5555/3495724.3496871} except in the accuracy parameter $\delta$.
In particular, we show the following theorem.
\begin{theorem}\label{thm: Rejection sampling from approximated distribution}
    Under the setting described above, given $0<\epsilon\leq \min\left\{Q_{\max}^{(\tau)},\sqrt{Q_{\max}^{(\tau)}}\right\}$, $\delta \in (0,1)$, $Q^{(\tau)}_{\max} > 0$ in Eq.~\eqref{eq: Q_max}, and access to the oracles $\mathrm{Q}(\bm Q^{(\tau)})$ and $\mathrm{SQ}(\sqrt{\hat{\bm q}^{(\rho)}})$, there exists a randomized classical algorithm that samples from the distribution ${\mathcal{D}}'$ whose total variation distance from the optimized probability distribution $\mathcal{D}_\epsilon$ in Definition~\ref{def: Optimized probability distribution} is at most $\delta$, i.e.,
    \begin{align}
        \dtv(\mathcal{D}_\epsilon,\mathcal{D}') \leq \delta,
    \end{align}
    with runtime
    \begin{equation}
        \begin{aligned}
            T_{\text{total}} &= O\left(D\;\polylog(G) + \mathbf{q}(\bm Q^{(\tau)}) + \mathbf{sq}(\sqrt{\hat{\bm q}^{(\rho)}}) \right) \\
            &\quad \times \tO\left(\left(\frac{Q^{(\tau)}_{\max}}{\epsilon}\right)^{18} \frac{1}{\delta^6}\right).
         \end{aligned}
    \end{equation}
\end{theorem}
The key difficulty of this sampling task is to compute the inverse of the inner matrix in Definition~\ref{def: Optimized probability distribution}
\begin{equation}\label{eq: inner matrix in optimized probability distribution}
    \left(\epsilon \bm I + \sqrt{\hat{\bm q}^{(\rho)}}\,\bm k\,\sqrt{\hat{\bm q}^{(\rho)}}\right)^{-1}.
\end{equation}
The main idea of our dequantization is to use Theorem~\ref{thm: type1} to compute the matrix in Eq.~\eqref{eq: inner matrix in optimized probability distribution}.
However, applying Theorem~\ref{thm: type1} requires an efficient implementation of the noisy query oracle $\mathrm{Q}_{\epsilon_k,\delta_k}(\bm k)$, which is not supplied directly.
Therefore, we show that we can construct an algorithm that acts as a query access oracle to a matrix $\widetilde{\bm k}$ such that $|\bm k(\alpha,\beta) - \widetilde{\bm k}(\alpha,\beta)|$ is small with a high probability by using $\mathrm{Q}(\bm Q^{(\tau)})$.
Then, we can apply Theorem~\ref{thm: type1} to approximate a matrix $(\epsilon \bm I + \sqrt{\hat{\bm q}^{(\rho)}}\,\bm k\,\sqrt{\hat{\bm q}^{(\rho)}})^{-1}$, by using $\mathrm{Q}(\bm Q^{(\tau)})$ and $\mathrm{SQ}(\sqrt{\hat{\bm q}^{(\rho)}})$, and thus, we can compute the optimized probability distribution $\mathcal{D}_\epsilon$ approximately and sample from it by using rejection sampling.
The details of the algorithm and the proof of Theorem~\ref{thm: Rejection sampling from approximated distribution} are given in the next appendix.

The runtime of the QML algorithm in~\eqref{eq: quantum runtime} depends only polylogarithmically on $1/\delta$, whereas that of our algorithm in Theorem~\ref{thm: Rejection sampling from approximated distribution} depends polynomially.
But the information extraction from quantum states basically requires polynomial costs in the desired accuracy.
This problem is discussed in the quantum-inspired dequantization field~\cite{10.1145/3313276.3316310,chia2022sampling}.

\section{Proof of Theorem~\ref{thm: Rejection sampling from approximated distribution}\label{sec: proof of Rejection sampling from approximated distribution}}
In this appendix, we give the proof of Theorem~\ref{thm: Rejection sampling from approximated distribution} by using Algorithm~\ref{alg: Rejection sampling from approximated distribution} with Theorem~\ref{thm: type1}.
The main idea of our algorithm consists of using Theorem~\ref{thm: type1} to compute the matrix in Eq.~\eqref{eq: inner matrix in optimized probability distribution} approximately, and then using the rejection sampling method to sample from a distribution close to $\mathcal{D}_\epsilon$.
However, the query access to the kernel matrix $\bm k$ is not given.
Therefore, we first show that we can approximately compute each entry of $\bm k$ by using $\mathrm{Q}(\bm Q^{(\tau)})$ with a high probability, and thus, we can construct a query access oracle to a matrix $\widetilde{\bm{k}}$ such that $|\bm{k}(\alpha,\beta) - \widetilde{\bm{k}}(\alpha,\beta)|$ is small with a high probability.

The matrix $\bm{k}$ in Eq.~\eqref{eq: k} is defined as the symmetric matrix, i.e., $\bm{k}(\alpha,\beta) = \bm{k}(\beta,\alpha)$ for $\alpha,\beta \in \mathcal{X}$, where
\begin{align}
    \bm{k}(\alpha,\beta) &= \Braket{\alpha \left|\bm{F}_D^\dagger \bm Q^{(\tau)}\bm{F}_D\right| \beta}\\
    &= \frac{1}{G^D}\sum_{y\in \mc{X}}Q^{(\tau)}(y)\omega^{y\cdot(\alpha - \beta)},
\end{align}
where $\omega = \exp\left(-\frac{2\pi i}{G}\right)$ is the primitive $G$-th root of unity.
Therefore, we have
\begin{align}
    \bm{k}(\alpha,\beta) &= \frac{1}{2}\left(\bm{k}(\alpha,\beta) + \bm{k}(\beta,\alpha)\right)\\
    &= \frac{1}{G^D}\sum_{y\in \mc{X}}Q^{(\tau)}(y)\frac{\omega^{y\cdot(\alpha - \beta)} + \omega^{y\cdot(\beta - \alpha)}}{2}\\
    &= \frac{1}{G^D}\sum_{y\in \mc{X}}Q^{(\tau)}(y)\cos\left(\frac{2\pi}{G}y\cdot(\alpha - \beta)\right).
\end{align}

In the following, for simplicity of analysis, we omit the time complexity of the precision term in computing the cosine function $\cos\left(\frac{2\pi}{G} g \right)$ for each entry $g \in \{0,1,\dots, G-1\}$, i.e., in fact, it takes $\polylog(G) \times \tO(\log(1/\epsilon))$ time to compute $\cos\left(\frac{2\pi}{G} g \right)$ with an additive error at most $\epsilon$ using the methods~\cite{brent1976fast,haible1998fast,brent2010modern}, but we alternatively use the time complexity $\polylog(G)$ since the time overhead is polylogarithmic in $Q^{(\tau)}_{\max}/\epsilon$, and thus we let its overhead be absorbed into the $\tO$ notation.

For each entry in the matrix $\bm{k}(\alpha, \beta)$, we estimate it by using the uniformly random sampling over $\mc{X}$ and the oracle $\mathrm{Q}(\bm Q^{(\tau)})$, and the error and the success probability are bounded using Hoeffding's inequality in Lemma~\ref{lemma: Hoeffding's inequality} as follows.
\begin{lemma}\label{lemma: matrix element estimation}
    Given the access to the oracle $\mathrm{Q}(\bm Q^{(\tau)})$, the parameters $0 < \epsilon, \delta < 1$, and indices $\alpha, \beta \in \mc{X}$, there exists a randomized classical algorithm that implements $\mathrm{Q}(\widetilde{\bm{k}})$ of $\widetilde{\bm{k}}$ such that the value $\widetilde{\bm{k}}(\alpha, \beta)$ satisfies
    \begin{align}
        \left|\bm{k}(\alpha,\beta) - \widetilde{\bm{k}}(\alpha, \beta)\right| \leq \epsilon,
    \end{align}
    with probability at least $1 - \delta$ within the time complexity
    \begin{equation}
    O\left(D \; \polylog(G) + \mathbf{q}(\bm Q^{(\tau)}) \right) \times\tO \left(\left(\frac{ Q^{(\tau)}_\text{max} }{\epsilon}\right)^2 \ln\left(\frac{1}{\delta}\right)\right)
    \end{equation}
\end{lemma}
\begin{proof}[Proof of Lemma~\ref{lemma: matrix element estimation}]
    The algorithm samples uniformly and returns the empirical average as follows:
    \begin{enumerate}
        \item Set $I \coloneq \left\lceil 2 \left(\frac{ \left(Q^{(\tau)}_\text{max}\right)^2 }{\epsilon^2}\right) \ln\left(\frac{2}{\delta}\right) \right\rceil$.
        \item Sample $y_1, \ldots, y_I \in \mc{X}$ uniformly at random from $\mc{X}$.
        \item Compute \begin{align}\label{eq: tilde_k definition}
            \widetilde{\bm{k}}(\alpha, \beta) &\coloneq \frac{1}{I}\sum_{i = 1}^{I}Q^{(\tau)}(y_i)\cos\left(\frac{2\pi}{G}y_i\cdot(\alpha - \beta)\right).
        \end{align}
        \item Return $\widetilde{\bm{k}}(\alpha, \beta)$.
    \end{enumerate}
    The correctness of the algorithm follows from Hoeffding's inequality in Lemma~\ref{lemma: Hoeffding's inequality}.
    The expectation value of the random variable $Q^{(\tau)}(y)\cos\left(\frac{2\pi}{G}y\cdot(\alpha - \beta)\right)$ for a uniform distribution over $\mc{X}$ is equal to $\bm{k}(\alpha, \beta)$, i.e.,
    \begin{align}
        &\mathbb{E}_{y \sim \mc{X}}\left[Q^{(\tau)}(y)\cos\left(\frac{2\pi}{G}y\cdot(\alpha - \beta)\right)\right] \\
        &\quad= \frac{1}{G^D}\sum_{y\in \mc{X}}Q^{(\tau)}(y)\cos\left(\frac{2\pi}{G}y\cdot(\alpha - \beta)\right)\\
        &\quad= \bm{k}(\alpha, \beta),
    \end{align}
    and the range of the random variable is bounded by $Q^{(\tau)}(y)\cos\left(\frac{2\pi}{G}y\cdot(\alpha - \beta)\right) \in [-Q^{(\tau)}_\text{max}, Q^{(\tau)}_\text{max}]$.
    Therefore, the absolute error of the estimation $\widetilde{\bm{k}}(\alpha, \beta)$ is bounded by $\epsilon$ with probability at least $1 - \delta$ by Hoeffding's inequality in Lemma~\ref{lemma: Hoeffding's inequality} as follows,
    \begin{align}
        &\Pr\left(\left|\widetilde{\bm{k}}(\alpha, \beta) - \bm{k}(\alpha, \beta)\right| \geq \epsilon\right) \\
        &\quad \leq 2\exp\left(-\frac{I\epsilon^2}{2(Q^{(\tau)}_\text{max})^2}\right)\\ \notag
        &\quad \leq 2\exp\left(-\ln\left(\frac{2}{\delta}\right)\right)= \delta.
    \end{align}

    The time complexity of the algorithm is dominated by the two parts: first, sampling $I$ indices from $\mc{X}$ in $O(\log(|\mc{X}|)) \times \tO(I)$ time, second, computing and averaging over $I$ samples in $\tO(I)\times O( (\mathbf{q}(\bm Q^{(\tau)}) + D + \polylog(G)))$, where tilde notation hides logarithmic factors arising from the computation of the cosine function.
    These lead to an overall time complexity of 
    \begin{equation}
        \begin{aligned}
             &O(I \log(|\mc{X}|)) + O(I (\mathbf{q}(\bm Q^{(\tau)}) + D + \polylog(G))) \\ 
             &\quad = O\left(D \polylog(G) + \mathbf{q}(\bm Q^{(\tau)}) \right) \\
             &\qquad \times\tO \left(\left(\frac{ Q^{(\tau)}_\text{max} }{\epsilon}\right)^2 \ln\left(\frac{1}{\delta}\right)\right),
        \end{aligned}
    \end{equation}
    which completes the proof of Lemma~\ref{lemma: matrix element estimation}.
\end{proof}

\begin{algorithm}[tbp]
\SetAlgoLined
\caption{Sampling from the optimized probability distribution~\label{alg: Rejection sampling from approximated distribution}}
\KwIn{Precision parameter $\delta \in (0,1)$, values $0<\epsilon\leq \min\left\{Q_{\max}^{(\tau)},\sqrt{Q_{\max}^{(\tau)}}\right\}$, $\bm k(0,0) = \Omega(1)$, and $Q^{(\tau)}_{\max} > 0$ in Eq.~\eqref{eq: Q_max}. Access to the oracles $\mathrm{Q}(\bm Q^{(\tau)})$ and $\mathrm{SQ}(\sqrt{\hat{\bm q}^{(\rho)}})$.}
\KwOut{A random variable $x \in \mc{X}$ sampled from a probability distribution ${\mc{D}}'$ within total variation distance $\delta$ of the distribution $\mc{D}_\epsilon$ in Definition~\ref{def: Optimized probability distribution}, i.e., $\dtv(\mc{D}', \mc{D}_\epsilon) \leq \delta$.}
Set
\begin{equation}\label{eq: error setup}
    \begin{aligned}
        I &= \left\lceil\frac{Q^{(\tau)}_{\max}}{\bm{k}(0,0)} \frac{12(Q_{\max}^{(\tau)} + \epsilon)^2 + \delta\epsilon^2}{12\epsilon(Q_{\max}^{(\tau)} + \epsilon - \delta\epsilon)} \ln\left(\frac{3}{\delta}\right)\right\rceil, \\
        \epsilon_k &= \min\left\{\frac{\epsilon}{2}, \frac{\delta\epsilon^2}{24(Q_{\max}^{(\tau)}+ \epsilon)}\right\}, \quad \delta_k = \frac{\delta\epsilon_D^4}{3(1 + \epsilon_D^2)},\\
        \epsilon_D &= \frac{\delta\epsilon^2}{48Q_{\max}^{(\tau)}(Q_{\max}^{(\tau)}+ \epsilon)} \quad \delta_D =\frac{\delta}{6},
    \end{aligned}
\end{equation}
$x = 0$, and $A = 0$\;
Apply Algorithm~\ref{alg: type1} with $\bm D = \sqrt{\hat{\bm q}^{(\rho)}}$, $\bm A=\bm k$, $\gamma=\epsilon$, and $f(x)=x^{-1}$, 
using the truncation parameters $(\epsilon_D, \delta_D)$ and the approximate query oracle $\mathrm{Q}_{\epsilon_k,\delta_k}({\bm{k}})$ obtained from Lemma~\ref{lemma: matrix element estimation}, and obtain two oracles $\mathrm{Q}_{\mathrm{sp}}(\bm N)$ and $\mathrm{SQ}(\tilde{\bm D})$\;
\For{$i \in [I]$}{
    Sample $y$ uniformly from $\mathcal{X}$ and set $x = y$\;
    Compute $\tilde{w}(x) $ in~Eq.~\eqref{eq: tilde w}\;
    \If{$\tilde{w}(x) < 0$ or $\tilde{w}(x) > \frac{Q_{\max}^{(\tau)}}{G^D} \frac{12(Q_{\max}^{(\tau)} + \epsilon)}{12\epsilon(Q_{\max}^{(\tau)} + \epsilon - \delta\epsilon)}$}{\Break\;}
Sample $u$ from the uniform distribution over $[0,1]$\;
\If{$u \leq \tilde{w}(x)\times \frac{G^D}{Q_{\max}^{(\tau)}} \frac{12\epsilon(Q_{\max}^{(\tau)} + \epsilon - \delta\epsilon)}{12(Q_{\max}^{(\tau)} + \epsilon)}$}{Set $A = 1$\;
\Break\;}}
\If{$A = 0$}{Uniformly sample $y$ from $\{0,\ldots, G-1\}^D$ and set $x = y$\;}
\Return $x$
\end{algorithm}

\begin{theorem}[The detailed restatement of Theorem~\ref{thm: Rejection sampling from approximated distribution}\label{thm: optimized random features sampling restated}]
    Given $0<\epsilon\leq \min\left\{Q_{\max}^{(\tau)},\sqrt{Q_{\max}^{(\tau)}}\right\}$, $\delta \in (0,1)$, $\bm k(0,0) = \Omega(1)$, $Q^{(\tau)}_{\max} > 0$ in Eq.~\eqref{eq: Q_max}, and access to the oracles $\mathrm{Q}(\bm Q^{(\tau)})$ and $\mathrm{SQ}(\sqrt{\hat{\bm q}^{(\rho)}})$, Algorithm~\ref{alg: Rejection sampling from approximated distribution} samples from the distribution ${\mathcal{D}}'$ whose total variation distance from the optimized probability distribution $\mathcal{D}_\epsilon$ in Definition~\ref{def: Optimized probability distribution} is at most $\delta$, i.e.,
    \begin{align}
        \dtv(\mathcal{D}', \mathcal{D}_\epsilon) \leq \delta,
    \end{align}
    with runtime
    \begin{equation}
        \begin{aligned}
            &O(D \; \polylog(G) + \mathbf{sq}\left(\sqrt{\hat{\bm q}^{(\rho)}}\right) + \mathbf{q}(\bm Q^{(\tau)}))  \\
            &\quad \times \tO((\frac{Q_{\max}^{(\tau)}}{\epsilon})^{18} \frac{1}{\delta^6}).
        \end{aligned}
    \end{equation}
\end{theorem}
\begin{proof}
    We first prove the correctness of Algorithm~\ref{alg: Rejection sampling from approximated distribution}, and then analyze the time complexity of Algorithm~\ref{alg: Rejection sampling from approximated distribution}.

    Let $\widetilde{\bm{k}}$ be a matrix defined from Lemma~\ref{lemma: matrix element estimation} with inputs $(\epsilon_k,\delta_k)$, i.e., for $\alpha,\beta \in \mc{X}$, $\widetilde{\bm{k}}$ satisfies
    \begin{equation}
        |\bm{k}(\alpha,\beta) - \widetilde{\bm{k}}(\alpha, \beta)| \leq \epsilon_k,
    \end{equation}
    with probability at least $1 - \delta_k$.
    
    On the successful event, we have
    \begin{equation}\label{eq: q k - tilde k q}
        \ope{\sqrt{\hat{\bm q}^{(\rho)}}(\bm{k}-\widetilde{\bm{k}})\sqrt{\hat{\bm q}^{(\rho)}}} \leq \epsilon_k \flo{\sqrt{\hat{\bm q}^{(\rho)}}}^2 = \epsilon_k,
    \end{equation}
    by setting $\widetilde{\bm{k}}(\alpha, \beta) = \widetilde{\bm{k}}(\beta, \alpha)$.
    The minimum eigenvalue of $\epsilon \bm I + \sqrt{\hat{\bm q}^{(\rho)}}\widetilde{\bm{k}}\sqrt{\hat{\bm q}^{(\rho)}}$ is at least
    \begin{equation}\label{eq: minimum tilde k}
        \begin{aligned}
            &\lambda_{\min}(\epsilon I + \sqrt{\hat{\bm q}^{(\rho)}}\widetilde{\bm{k}}\sqrt{\hat{\bm q}^{(\rho)}}) \\
            &\quad \geq \lambda_{\min}(\bm{\Sigma}_\epsilon) - \ope{\sqrt{\hat{\bm q}^{(\rho)}}(\bm{k}-\widetilde{\bm{k}})\sqrt{\hat{\bm q}^{(\rho)}}} \\
            &\quad\geq \lambda_{\min}(\bm{\Sigma}_\epsilon) - \epsilon_k \\
            &\quad \geq \epsilon - \epsilon_k,
        \end{aligned}
    \end{equation}
    where the first inequality follows from Weyl's inequality, the second inequality follows from Eq.~\eqref{eq: q k - tilde k q}, and the last inequality follows from the fact that $\sqrt{\hat{\bm q}^{(\rho)}}\bm{k}\sqrt{\hat{\bm q}^{(\rho)}}$ is positive semidefinite. 
    
    Algorithm~\ref{alg: Rejection sampling from approximated distribution} first applies Algorithm~\ref{alg: type1} with $\mathrm{SQ}(\sqrt{\hat{\bm q}^{(\rho)}})$, $\mathrm{Q}_{\epsilon_k,\delta_k}(\bm{k})$, $\delta_D \in (0,1]$, and $\epsilon_D, \epsilon > 0$, and obtains a query oracle $\qsp(\bm N)$ for a real symmetric matrix $\bm N \in \mathbb{R}^{n \times n}$ such that
    \begin{equation}\label{eq: opnorm epsilon-N}
    \begin{aligned}
        &\ope{\bm{\Sigma}_\epsilon^{-1} - \bm N} \\
        &\quad\leq \frac{1}{(\epsilon - \epsilon_k)^2} \left( 2\ope{\bm k}\ope{\sqrt{\hat{\bm q}^{(\rho)}}}\epsilon_D  + \epsilon_k \flo{\sqrt{\hat{\bm q}^{(\rho)}}}^2 \right) \\
        &\quad \leq \frac{1}{(\epsilon - \epsilon_k)^2} \left( 2Q_{\max}^{(\tau)}\epsilon_D + \epsilon_k \right),
    \end{aligned}
    \end{equation}
    where the first inequality follows by applying Theorem~\ref{thm: type1} with $f(x) = x^{-1}$, but using the inverse-specific operator-norm perturbation bound instead of the general Lipschitz bound in~\eqref{eq: the precision of type1}.
    Indeed, on matrices $\bm{X}$ and $\bm{Y}$ whose spectra are contained in $[a,\infty)$ for $a > 0$, the inverse map satisfies
    \begin{equation}
        \ope{\bm X^{-1} - \bm Y^{-1}} \leq \frac{1}{a^2}\ope{\bm X - \bm Y}.
    \end{equation}
    Here, $a = \epsilon - \epsilon_k$ by Eq.~\eqref{eq: minimum tilde k}.
    Thus, the general factor $C(1 + \log(1 + s))L$ in~\eqref{eq: the precision of type1} is replaced by $1/(\epsilon - \epsilon_k)^2$.
    The second inequality in~\eqref{eq: opnorm epsilon-N} follows from $\ope{\sqrt{\hat{\bm q}^{(\rho)}}} \leq 1 ,\flo{\sqrt{\hat{\bm q}^{(\rho)}}}^2 = 1$, and $\ope{\bm k} \leq Q_{\max}^{(\tau)}$.

    Let $w: \mathcal{X} \to \mathbb{R}$ be a weight function
    \begin{equation}
        w(x) =  \Braket{x \left| \sqrt{\bm Q^{(\tau)}}\,\bm{F}_D^\dagger \sqrt{\hat{\bm q}^{(\rho)}} \bm{\Sigma}_\epsilon^{-1} \sqrt{\hat{\bm q}^{(\rho)}}\,\bm{F}_D\sqrt{\bm Q^{(\tau)}} \right| x},
    \end{equation}
    and define $W \coloneq \sum_{x \in \mathcal{X}} w(x)$.
    For every $x \in \mathcal{X}$, we have
    \begin{equation}\label{eq: w(x) is larger}
        \begin{aligned}
            w(x) &\geq \Braket{x \left|\sqrt{\bm Q^{(\tau)}}\bm{F}_{D}^\dagger {\hat{\bm q}^{(\rho)}} \bm{F}_{D}\sqrt{\bm Q^{(\tau)}}\right| x}  \lambda_{\max}\left(\bm{\Sigma}_\epsilon\right)^{-1} \\
            &\geq \frac{1}{(Q_{\max}^{(\tau)} + \epsilon)} \Braket{x \left|\sqrt{\bm Q^{(\tau)}}\bm{F}_{D}^\dagger {\hat{\bm q}^{(\rho)}} \bm{F}_{D}\sqrt{\bm Q^{(\tau)}}\right| x}.
        \end{aligned}
    \end{equation}
    Then, we can write the optimized distribution $\mathcal{D}_\epsilon$ in Definition~\ref{def: Optimized probability distribution} as 
    \begin{equation}
        \mathcal{D}_\epsilon(x) = \frac{w(x)}{W}.
    \end{equation}
    Also, by using the approximated weight function $\tilde{w}$
    \begin{equation}\label{eq: tilde w}
        \tilde{w}(x) \coloneq \Braket{x \left| \sqrt{\bm Q^{(\tau)}}\,\bm{F}_D^\dagger \sqrt{\hat{\bm q}^{(\rho)}} \bm N \sqrt{\hat{\bm q}^{(\rho)}}\,\bm{F}_D\sqrt{\bm Q^{(\tau)}} \right| x},
    \end{equation}
    and its summation $\tilde{W} = \sum_{x \in \mathcal{X}}\tilde{w}(x)$, we define the approximated distribution $\tilde{\mathcal{D}}$ as 
    \begin{equation}\label{eq: approximated probability distirbution}
        \tilde{\mathcal{D}}(x) = \frac{\tilde w(x)}{\tilde W},
    \end{equation}
    where Algorithm~\ref{alg: Rejection sampling from approximated distribution} tries to sample from the distribution $\tilde{\mc D}$.

    We analyze the differences between $w(x)$ and $\tilde{w}(x)$, and $W$ and $\tilde{W}$,
    \begin{equation}\label{eq: w - tilde w}
        \begin{aligned}
            &|w(x) - \tilde{w}(x)| \\
            &\quad =  \left| \left\langle x \middle|  \sqrt{\bm Q^{(\tau)}}\,\bm{F}_D^\dagger \sqrt{\hat{\bm q}^{(\rho)}} \left(\bm{\Sigma}_\epsilon^{-1} - \bm{N} \right) \right.\right.\\
            &\qquad \times \left. \left. \sqrt{\hat{\bm q}^{(\rho)}}\,\bm{F}_D\sqrt{\bm Q^{(\tau)}} \middle| x \right\rangle \right| \\ 
            &\quad \leq \left\langle x \middle| \sqrt{\bm Q^{(\tau)}}\,\bm{F}_D^\dagger {\hat{\bm q}^{(\rho)}} \bm{F}_D\sqrt{\bm Q^{(\tau)}} \middle| x \right\rangle \ope{\bm{\Sigma}_\epsilon^{-1} - \bm{N}}\\
            &\quad \leq w(x)\frac{Q_{\max}^{(\tau)} + \epsilon}{(\epsilon - \epsilon_k)^2} \left( 2Q_{\max}^{(\tau)}\epsilon_D + \epsilon_k \right),
        \end{aligned}
    \end{equation}
    where the final line follows from Equations~\eqref{eq: opnorm epsilon-N} and~\eqref{eq: w(x) is larger},
    \begin{equation}
        \begin{aligned}
            |W - \tilde{W}| &= \left| \sum_{x \in \mathcal{X}}w(x) - \sum_{x \in \mathcal{X}}\tilde{w}(x)\right| \\
            &\leq \sum_{x \in \mathcal{X}} \left|w(x) - \tilde{w}(x)\right| \\
            &\leq \sum_{x\in \mathcal{X}} w(x) \frac{Q_{\max}^{(\tau)} + \epsilon}{(\epsilon - \epsilon_k)^2} \left( 2Q_{\max}^{(\tau)}\epsilon_D + \epsilon_k \right) \\
            &\leq W \frac{Q_{\max}^{(\tau)} + \epsilon}{(\epsilon - \epsilon_k)^2} \left( 2Q_{\max}^{(\tau)}\epsilon_D + \epsilon_k \right),
        \end{aligned}
    \end{equation}
    where we use Eq.~\eqref{eq: w - tilde w} in the third line.

    By the parameter choices in Eq.~\eqref{eq: error setup}, we have the following inequalities
    \begin{align}
        \frac{1}{(\epsilon - \epsilon_k)^2} \leq \frac{4}{\epsilon^2}, \quad 2Q_{\max}^{(\tau)}\epsilon_D + \epsilon_k \leq \frac{\delta\epsilon^2}{12(Q_{\max}^{(\tau)} + \epsilon)},
    \end{align}
    and thus, it gives 
    \begin{align}
        \frac{Q_{\max}^{(\tau)} + \epsilon}{(\epsilon - \epsilon_k)^2} \left( 2Q_{\max}^{(\tau)}\epsilon_D + \epsilon_k \right) \leq \frac{\delta}{3}.
    \end{align}
    Now, we bound the total variation distance between the two distributions $\mathcal{D}_\epsilon$ and $\tilde{\mathcal{D}}$ as follows,
    \begin{equation}
        \begin{aligned}
            \dtv(\mathcal{D}_\epsilon,\tilde{\mathcal{D}}) &= \frac{1}{2} \sum_{x \in \mathcal{X}}\left| \frac{w(x)}{W} - \frac{\tilde{w}(x)}{\tilde{W}} \right|\\
            &= \frac{1}{2}\sum_{x \in \mathcal{X}}\left| \frac{w(x) - \tilde{w}(x)}{W} - \tilde{w}(x)\left(\frac{1}{\tilde{W}} - \frac{1}{{W}}\right)\right|\\
            &\leq \frac{1}{2}\sum_{x \in \mathcal{X}} \left\{ \frac{|w(x) - \tilde{w}(x)|}{W} + \tilde{w}(x) \frac{|W - \tilde{W}|}{W\tilde{W}}\right\} \\ 
            &= \frac{1}{2}\sum_{x \in \mathcal{X}} \left\{ \frac{|w(x) - \tilde{w}(x)|}{W}\right\} + \frac{|W - \tilde{W}|}{2W} \\ 
            & \leq \frac{(Q_{\max}^{(\tau)} + \epsilon)}{(\epsilon - \epsilon_k)^2} \left( 2Q_{\max}^{(\tau)}\epsilon_D + \epsilon_k \right)\\
            &\leq \frac{\delta}{3}.
        \end{aligned}
    \end{equation}

    Then, we analyze the probability Algorithm~\ref{alg: Rejection sampling from approximated distribution} succeeds in sampling from $\tilde{\mc D}$, i.e., the probabilities that Algorithm~\ref{alg: type1} and the rejection sampling procedure succeed.
    The success probability of Algorithm~\ref{alg: type1} is at least
    \begin{equation}\label{eq: algorithm1 success probability}
        1 - (\delta_D + \frac{1 + \epsilon_D^2}{2\epsilon_D^4} \delta_k) \geq 1 - \frac{\delta}{3},
    \end{equation}
    by Theorem~\ref{thm: type1}.
    Before analyzing the rejection sampling part, we give the upper bound $\tilde{w}_{\max}$ of the function $\tilde{w}(x)$ in~\eqref{eq: tilde w} over all indices $x \in \mc{X}$, which is used for rejection sampling,
    \begin{align}
\label{eq: tilde_d_max}
    \tilde{w}(x)
    &= \Braket{x \left| \sqrt{\bm Q^{(\tau)}}\,\bm{F}_D^\dagger \sqrt{\hat{\bm q}^{(\rho)}} \bm N \sqrt{\hat{\bm q}^{(\rho)}}\,\bm{F}_D\sqrt{\bm Q^{(\tau)}} \right| x} \\
    &\leq \Braket{x \left|\sqrt{\bm Q^{(\tau)}}\bm{F}_{D}^\dagger {\hat{\bm q}^{(\rho)}} \bm{F}_{D}\sqrt{\bm Q^{(\tau)}}\right| x} \ope{\bm N} \\
    &\leq \frac{Q^{(\tau)}(x)}{G^D} \frac{1}{\epsilon - (2Q_{\max}^{(\tau)}\epsilon_D + \epsilon_k)} \\
    &\leq \frac{Q_{\max}^{(\tau)}}{G^D} \frac{12(Q_{\max}^{(\tau)} + \epsilon)}{12\epsilon(Q_{\max}^{(\tau)} + \epsilon - \delta\epsilon)}.
\end{align}
    For simplicity, we denote the upper bound of $\tilde{w}(x)$ as $\tilde{w}_{\max}$, i.e.,
    \begin{equation}\label{eq: tilde w max}
        \tilde{w}_{\max} = \frac{Q_{\max}^{(\tau)}}{G^D} \frac{12(Q_{\max}^{(\tau)} + \epsilon)}{12\epsilon(Q_{\max}^{(\tau)} + \epsilon - \delta\epsilon)}.
    \end{equation}
    Also, we compute the lower bound of the sum $\tilde{W} = \sum_{x}\tilde{w}(x)$ of the function $\tilde{w}(x)$ in~\eqref{eq: tilde w} over $x\in\mathcal{X}$ by analyzing the lower bound of $\tilde{w}(x)$ except for $Q^{(\tau)}(x)$,
    \begin{equation}\label{eq: tilde_d lower bound}
    \begin{aligned}
        \tilde{w}(x) &\geq \frac{Q^{(\tau)}(x)}{G^D} \frac{1}{Q^{(\tau)}_{\max} + \epsilon + (2Q_{\max}^{(\tau)}\epsilon_D + \epsilon_k)}\\
        &\geq \frac{Q^{(\tau)}(x)}{G^D} \frac{12(Q_{\max}^{(\tau)} + \epsilon)}{12(Q_{\max}^{(\tau)} + \epsilon)^2 + \delta\epsilon^2},
    \end{aligned}
\end{equation}
    and therefore we can compute the sum $\tilde{W}$ of the function $\tilde{w}(x)$ in~\eqref{eq: tilde w} over $x\in \mathcal{X}$ as follows,
    \begin{equation}\label{eq: tilde_d sum}
    \begin{aligned}
        \sum_{x}\tilde{w}(x) &\geq \frac{\sum_{x} Q^{(\tau)}(x)}{G^D} \frac{12(Q_{\max}^{(\tau)} + \epsilon)}{12(Q_{\max}^{(\tau)} + \epsilon)^2 + \delta\epsilon^2} \\
        &= \bm{k}(0,0)\frac{12(Q_{\max}^{(\tau)} + \epsilon)}{12(Q_{\max}^{(\tau)} + \epsilon)^2 + \delta\epsilon^2}.
    \end{aligned}
\end{equation}
    Now, we analyze the rejection sampling in Algorithm~\ref{alg: Rejection sampling from approximated distribution}.
    First, we compute the acceptance probability of rejection sampling in Algorithm~\ref{alg: Rejection sampling from approximated distribution}.
    In Algorithm~\ref{alg: Rejection sampling from approximated distribution}, for each uniformly sampled $x \in \mc{X}$, we accept $x$ with probability $\tilde{w}(x)/\tilde{w}_{\max}$.
    The probability $\pr(x:\text{accept})$ of each $x \in \mc{X}$ is returned is given by
    \begin{equation}\label{eq: returned probabilityq}
        \begin{aligned}
            \pr(x:\text{accept}) &= \frac{1}{G^D}\times \frac{\tilde{w}(x)}{\tilde{w}_{\max}}.
        \end{aligned}
    \end{equation}
    Thus, the probability that a proposal is accepted in Algorithm~\ref{alg: Rejection sampling from approximated distribution} is given by summing the probability $\pr(x:\text{accept})$ over all $x \in \mc{X}$, i.e.,
    \begin{equation}\label{eq: acceptance probability}
        \begin{aligned}
            \pr(\text{accept}) &= \sum_{x \in \mc{X}} \pr(x:\text{accept})\\
            &= \frac{1}{G^D}\sum_{x \in \mc{X}} \frac{\tilde{w}(x)}{\tilde{w}_{\max}}\\
        \end{aligned}
    \end{equation}
    Due to Equations~\eqref{eq: returned probabilityq} and~\eqref{eq: acceptance probability}, the probability $\pr(x:\text{success})$ returned by Algorithm~\ref{alg: Rejection sampling from approximated distribution} conditioned on successful matrix approximation and acceptance is given by, for each $x \in \mc{X}$,
    \begin{equation}\label{eq: returned distribution}
        \begin{aligned}
            \pr(x:\text{success}) &= \frac{\pr(x:\text{accept})}{\pr(\text{accept})}\\
            &= \frac{\tilde{w}(x)}{\tilde{W}},
        \end{aligned}
    \end{equation}
    which is exactly the same as the distribution $\tilde{\mc{D}}$ in Eq.~\eqref{eq: approximated probability distirbution}.

    Next, we analyze the failure probability of the rejection sampling, namely, the probability that no proposal is accepted in $I$ iterations in Algorithm~\ref{alg: Rejection sampling from approximated distribution}.
    The probability that the rejection sampling fails is given by
    \begin{equation}\label{eq: failure probability of rejection sampling}
        \begin{aligned}
           &\pr(\text{failure})\\
             &\quad =  (1 - \pr(\text{accept}))^{I}\\
            &\quad\leq \exp\left(-I \times \pr(\text{accept})\right)\\
            &\quad = \exp\left(-I \times \frac{1}{G^D} \frac{\tilde{W}}{\tilde{w}_{\max}} \right)\\
            &\quad \leq  \exp\left(-I \times \frac{\bm{k}(0,0)}{G^D} \frac{12(Q_{\max}^{(\tau)} + \epsilon)}{12(Q_{\max}^{(\tau)} + \epsilon)^2 + \delta\epsilon^2} \frac{G^D}{Q_{\max}^{(\tau)}} \right.\\
            &\qquad \times \left. \frac{12\epsilon(Q_{\max}^{(\tau)} + \epsilon - \delta\epsilon)}{12(Q_{\max}^{(\tau)} + \epsilon)} \right)\\
            &\quad =  \exp\left(-I \times \frac{\bm{k}(0,0)}{Q^{(\tau)}_{\max}} \frac{12\epsilon(Q_{\max}^{(\tau)} + \epsilon - \delta\epsilon)}{12(Q_{\max}^{(\tau)} + \epsilon)^2 + \delta\epsilon^2}\right)\\
            &\quad=  \exp\left(-\left\lceil\frac{Q^{(\tau)}_{\max}}{\bm{k}(0,0)} \frac{12(Q_{\max}^{(\tau)} + \epsilon)^2 + \delta\epsilon^2}{12\epsilon(Q_{\max}^{(\tau)} + \epsilon - \delta\epsilon)} \ln\left(\frac{3}{\delta}\right) \right\rceil \right. \\
            &\qquad\left. \times \frac{\bm{k}(0,0)}{Q^{(\tau)}_{\max}} \frac{12\epsilon(Q_{\max}^{(\tau)} + \epsilon - \delta\epsilon)}{12(Q_{\max}^{(\tau)} + \epsilon)^2 + \delta\epsilon^2}\right)\\
            &\quad \leq  \exp\left(-\ln(\frac{3}{\delta})\right)\\
            &\quad = \frac{\delta}{3},\\
        \end{aligned}
    \end{equation}
    where we use Equations~\eqref{eq: tilde w max} and \eqref{eq: tilde_d sum} in the second inequality.

    Let $p_{\mathrm{fail}}$ be the probability that either Algorithm~\ref{alg: type1} fails or the rejection sampling fails.
    From the above bounds and the union bound, we have
    \begin{equation}
        p_{\mathrm{fail}} \leq \frac{2\delta}{3}.
    \end{equation}
    Let $\mathcal{F}$ denote the conditional output distribution on the failure event.
    If this failure occurs, Algorithm~\ref{alg: Rejection sampling from approximated distribution} samples from this distribution $\mathcal{F}$.
    Thus, the output distribution $\mathcal{D}'$ sampled by Algorithm~\ref{alg: Rejection sampling from approximated distribution} satisfies
    \begin{equation}\label{eq: total variation distance}
        \begin{aligned}
            \dtv(\mc{D}', \mc{D}_\epsilon) & = \dtv((1 - p_{\mathrm{fail}})\tilde{\mc{D}} + p_{\mathrm{fail}}\mc{F}, \mc{D}_\epsilon)\\
            &\leq (1 - p_{\mathrm{fail}})\dtv(\tilde{\mc{D}}, \mc{D}_\epsilon) + p_{\mathrm{fail}}\dtv(\mc{F}, \mc{D}_\epsilon)\\
            &\leq \frac{\delta}{3} + \frac{2\delta}{3} \\
            &= \delta,
        \end{aligned}
    \end{equation}
    which assures the correctness of Algorithm~\ref{alg: Rejection sampling from approximated distribution}.

    Finally, we analyze the time complexity of Algorithm~\ref{alg: Rejection sampling from approximated distribution}.
    The runtime of Algorithm~\ref{alg: Rejection sampling from approximated distribution} is maximal when the algorithm runs for all $I$ iterations without accepting a proposal.
    The dominant parts of the runtime in this case are the runtime of Algorithm~\ref{alg: type1} and the $I$ evaluations of the weight function $\tilde{w}(x)$ in~\eqref{eq: tilde w} and the sampling from uniform distribution.

    The runtime of Algorithm~\ref{alg: type1} in the present application is
    \begin{equation}
        O\left(\mathbf{sq}\left(\sqrt{\hat{\bm q}^{(\rho)}}\right)\frac{1}{\epsilon_D^2}\log\left(\frac{1}{\delta_D}\right) + \mathbf{q}_{\epsilon_k,\delta_k}(\bm k)\frac{1}{\epsilon_D^4} + \frac{1}{\epsilon_D^6}\right)
    \end{equation}
    where 
    \begin{equation}
        \begin{aligned}
            \mathbf{q}_{\epsilon_k,\delta_k}(k) &= O\left(D \; \polylog(G) + \mathbf{q}(\bm Q^{(\tau)}) \right) \\
            &\quad\times\tO \left(\left(\frac{ Q^{(\tau)}_\text{max} }{\epsilon_k}\right)^2 \ln\left(\frac{1}{\delta_k}\right)\right),
        \end{aligned}
    \end{equation}
    due to Lemma~\ref{lemma: matrix element estimation}, and the sampling from the uniform distribution takes $O(D\log(G))$ time by using $O(D\log(G))$ random bits.

    The non-trivial part is the evaluation of the weight function $\tilde{w}(x)$ in~\eqref{eq: tilde w} whose direct expansion gives 
    \begin{equation}
    \begin{aligned}
        &\frac{G^D}{Q^{(\tau)}(x)} \tilde{w}(x) \\
        &\quad= \sum_{\alpha, \beta \in \mc{X}} \sqrt{\hat {q}^{(\rho)}(\alpha){\hat q}^{(\rho)}(\beta)} \cos({\frac{2\pi  x \cdot (\alpha - \beta)}{G}}) \bm{N}(\alpha,\beta).
    \end{aligned}
    \end{equation}
    Naively, this expression requires a double summation over $\mc{X} \times \mc{X}$, which is infeasible since $|\mc{X}| = G^D$.
    However, since $\bm{N}$ has a block structure, denoting $S$ as the principal indices of $\bm{N}$, the cross terms between $S$ and $\mc{X}\setminus S$ vanish, i.e.,
    \begin{equation}
    \begin{aligned}
        &\sum_{\alpha, \beta \in S} \sqrt{{\hat q}^{(\rho)}(\alpha){\hat q}^{(\rho)}(\beta)} \cos({\frac{2\pi  x \cdot (\alpha - \beta)}{G}}) \bm{N}(\alpha,\beta)  \\
        & \qquad +  \sum_{\alpha, \beta \in \mathcal{X}\setminus S} \sqrt{{\hat q}^{(\rho)}(\alpha){\hat q}^{(\rho)}(\beta)} \cos({\frac{2\pi  x \cdot (\alpha - \beta)}{G}})\frac{\delta_{\alpha,\beta}}{\epsilon}\\
        &\quad= \sum_{\alpha, \beta \in S} \sqrt{{\hat q}^{(\rho)}(\alpha){\hat q}^{(\rho)}(\beta)} \cos({\frac{2\pi  x \cdot (\alpha - \beta)}{G}}) \bm{N}(\alpha,\beta)  \\
        &\qquad + \frac{1}{\epsilon}\sum_{\alpha \in \mathcal{X}\setminus S} \hat q^{(\rho)}(\alpha).
    \end{aligned}
    \end{equation}
    Here, although $\sum_{\alpha \in \mathcal{X}\setminus S} 
    \hat q^{(\rho)}(\alpha)$ has $|\mathcal X| - |S|$ terms, it can be computed by
    \begin{equation}
    \begin{aligned}
        \sum_{\alpha \in \mathcal{X}\setminus S} \hat q^{(\rho)}(\alpha) &= 1 - \sum_{\alpha \in S} \hat q^{(\rho)}(\alpha).
    \end{aligned}
    \end{equation}
    Therefore, we can compute the $\tilde{w}(x)$ by computing $Q^{(\tau)}(x)$ with $\mathrm{Q}(\bm Q^{(\tau)})$, the $s \times s$ summation with $\qsp(\bm N)$ and computing $\sum_{\alpha \in S} \hat q^{(\rho)}(\alpha)$ by using $\mathrm{SQ}({\widetilde{D}})$ obtained from Lemma~\ref{lem: low rank approximation of D} used in Algorithm~\ref{alg: type1}.
    Namely, the $I$-times computation of the approximated weight function $\tilde{w}(x)$ takes $O(I \times (D + \polylog(G) + s^2 \mathbf{q}(\bm N)))$,
    where $s \leq 1/\epsilon_D^2$ and $\mathbf{q}(\bm N) = O(\log(s))$ from Theorem~\ref{thm: type1}.

    Therefore, the total runtime of Algorithm~\ref{alg: Rejection sampling from approximated distribution} is upper bounded by 
    \begin{equation}\label{eq: total runtime}
        \begin{aligned}
           & O\left(\mathbf{sq}\left(\sqrt{\hat{\bm q}^{(\rho)}}\right)\frac{1}{\epsilon_D^2}\log\left(\frac{1}{\delta_D}\right)+ \mathbf{q}_{\epsilon_k,\delta_k}(\bm k)\frac{1}{\epsilon_D^4} + \frac{1}{\epsilon_D^6}\right) \\
           &\qquad+ I \times O(D\; \log(G)) \\
           & \qquad + O(I \times (D\log(G) + \polylog(G) + s^2 \mathbf{q}(\bm N))) \\
           &\quad= O(D \; \polylog(G) + \mathbf{sq}\left(\sqrt{\hat{\bm q}^{(\rho)}}\right) + \mathbf{q}(\bm Q^{(\tau)}))\\
           &\quad \times \tO((\frac{Q_{\max}^{(\tau)}}{\epsilon})^{18} \frac{1}{\delta^6}),
        \end{aligned}
    \end{equation}
    where we use $\bm k(0,0) = \Omega(1)$.
    This completes the proof of Theorem~\ref{thm: optimized random features sampling restated}.
\end{proof}

\end{document}